\documentclass[11pt]{article}
\usepackage{multirow}
\usepackage{changepage}

\usepackage[T1]{fontenc}
\usepackage{fourier}
\usepackage{fullpage}
\usepackage{amsmath}
\usepackage{amssymb}
\usepackage{amsthm}
\usepackage{subfigure}
\usepackage{comment}

\usepackage[noend]{algorithmic}
\usepackage{algorithm}
\usepackage{tikz}
\usepackage{pgfplots}
\usepackage{float}

\newtheorem{defn}{Definition}

\newcommand{\SAT}{\textsf{SAT}}
\newcommand{\sfAND}{\textsf{AND}}
\newcommand{\sfANDs}{\textsf{AND}s }
\newcommand{\sfOR}{\textsf{OR}}
\newcommand{\CNF}{\textsf{CNF}}
\newcommand{\true}{\textsf{true}}
\newcommand{\false}{\textsf{false}}
\newcommand{\unknown}{\textsf{unknown}}
\newcommand{\sfboth}{\textsf{both}}
\newcommand{\NP}{\textsf{NP}}
\newcommand{\SharpP}{\textsf{\#P}}
\newcommand{\Tetris}{\textsf{Tetris}}
\newcommand{\SIMD}{\textsf{SIMD}}
\newcommand{\DNF}{\textsf{DNF}}
\newcommand{\CNFTetris}{\textsf{CNFTetris}}
\newcommand{\SharpSAT}{\textsf{\#SAT}}
\newcommand{\DPLL}{\textsf{DPLL}}
\newcommand{\SNAP}{\textsf{SNAP}}
\newcommand{\Insert}{\textsf{Insert}}
\newcommand{\InsertCluster}{\textsf{InsertCluster}}
\newcommand{\ContainsCluster}{\textsf{ContainsCluster}}
\newcommand{\Contains}{\textsf{Contains}}
\newcommand{\GetAllContainingBoxes}{\textsf{GetAllContainingBoxes}}
\newcommand{\GetAllContainingBoxesCluster}{\textsf{GetAllContainingBoxesCluster}}
\newcommand{\boxes}{\textsf{BOXES}}
\newcommand{\children}{\textsf{CHILDREN}}
\newcommand{\depth}{\textsf{DEPTH}}
\newcommand{\SharpSATs}{\textsf{sharpSAT}}
\newcommand{\Cachet}{\textsf{Cachet}}
\newcommand{\dSharp}{\textsf{dSharp}}
\newcommand{\MaxSAT}{\textsf{MaxSAT}}

\newtheorem{theorem}{Theorem}
\newtheorem{lemma}[theorem]{Lemma}
\newcounter{exampleCounter}
\newtheorem{example}[exampleCounter]{Example}
\newcounter{questionCounter}
\newtheorem{question}[questionCounter]{Question}

\usepackage{url}

\title{Implementation of {\huge \Tetris} as a Model Counter\thanks{This research was supported in part by grant\# NSF CCF-1319402.}}

\author{
Jimmy Dobler \\
University at Buffalo, SUNY \\
jdobler@buffalo.edu 
\and
 Atri Rudra\\
University at Buffalo, SUNY \\
atri@buffalo.edu
}
\date{\vspace{-3ex}}

\begin{document}

\maketitle

\begin{abstract}
Solving \SharpSAT\ problems is an important area of work. In this paper, we discuss implementing \Tetris, an algorithm originally designed for handling natural joins, as an exact model counter for the \SharpSAT\ problem. \Tetris\ uses a simple geometric framework, yet manages to achieve the fractional hypertree-width bound. Its design allows it to handle complex problems involving extremely large numbers of clauses on which other state-of-the-art model counters do not perform well, yet still performs strongly on standard \SAT\ benchmarks. \par
    We have achieved the following objectives. First, we have found a natural set of model counting benchmarks on which \Tetris\ outperforms other model counters. Second, we have constructed a data structure capable of efficiently handling and caching all of the data \Tetris\ needs to work on over the course of the algorithm. Third, we have modified \Tetris\ in order to move from a theoretical, asymptotic-time-focused environment to one that performs well in practice. In particular, we have managed to produce results keeping us within a single order of magnitude as compared to other solvers on most benchmarks, and outperform those solvers by multiple orders of magnitude on others.  \par
    
\end{abstract}

\newpage

\section{Introduction}

\SharpSAT\ is the prototypical \SharpP-complete problem. \SharpSAT\ (as well as its \NP-complete cousin \SAT) are not only of great interest in computational complexity but by their completeness turn out to be a great tool to model a wide host of practical problems. This has led to an explosion of \SAT\ solvers that try to solve practical instances of \SharpSAT\ or \SAT\ by exploiting structure in these instances. For this paper, we will assume the importance of designing \SharpSAT\ and \SAT\ solvers as a given. We refer the reader to the book chapters by Gomes, Sabharwal and Selman on \SharpSAT\ solvers (also known as model counters)~\cite{GSS09} and by Gomes et al. on \SAT\ solvers~\cite{GKSS08} for more details.

A common technique is the  \DPLL\ procedure, a depth-first search procedure where the algorithm makes guesses on the assignments one variable at a time, determines at each stage whether or not this produces a conflict, and uses that information to learn new clauses and get closer to finding the satisfying assignment \cite{heras2008minimaxsat}.

Recently in the database literature, the work of Abo Khamis et al.~\cite{AboKhamis:2015} connected the \DPLL\ procedure to computing natural joins. In particular, they presented the \Tetris\ algorithm, which computes the natural join with beyond worst-case theoretical guarantees. As a special case, \Tetris\ also recovers some of the recent worst-case optimal join results~\cite{NPRR,V14,NRR13}. Abo Khamis et al. then showed that \Tetris\ is an \DPLL\ procedure and pointed out how one of the main step in their algorithm is exactly the {\em resolution} step that is ubiquitous in \DPLL -based \SAT\ solvers. Given the close ties of \SAT\ solvers to \DPLL, they left open the following intriguing possibility:
\begin{quote}
{\em Can \Tetris\ be implemented as a \SAT\ solver or model counter that can compete with state-of-the-art solvers?}
\end{quote}

\paragraph{Our contributions} Our main result in this paper is to show that \Tetris\ can indeed be implemented as a model counter that is competitive with state-of-the-art model counters on actual datasets.

While~\cite{AboKhamis:2015} presented a nice geometric framework to reason about algorithms to compute the natural join query, some of its simplicity arose from inefficiencies that matter when implementing \Tetris\ as a model counter. Before we present the issues we tackle, we give a quick overview of \Tetris. The fundamental idea is that, rather than working to create the output of a join directly, it instead attempts to rule out large sections of the cross product of the joined tables \cite{AboKhamis:2015}. Initially, \Tetris\ is given a set of sets whose union is the set of all incorrect solutions to the problem \Tetris\ is solving. In other words, any solution to the problem must not be a member of this union. By efficiently querying this set of sets, and by adding to it intelligently at various times (such as by adding a new exclusion whenever an output point is found), \Tetris\ is able to rule out increasingly large sets of potential solutions. Once it has ruled out all possible solutions, it terminates and outputs the list of solutions.

We tackle the following three issues with the theoretical presentation of \Tetris\ in~\cite{AboKhamis:2015}:
\begin{enumerate}
\item At any point, \Tetris\ needs to keep track of the union of all the potential solutions it has ruled out. To do this,~\cite{AboKhamis:2015} used a simple trie data structure to keep track of the union. However, this loses some poly-logarithmic factors and proves detrimental to practical performance. To deal with this, we design a new data structure that essentially compresses consecutive layers in the traditional trie into one `mega layer.' Inspired by the used of SIMD instructions by EmptyHeaded~\cite{emptyheaded} (to speed up implementations of worst-case optimal join algorithms~\cite{NRR13}), we set up the compression in a manner that lends itself to speedup via SIMD instructions.
\item The analysis of \Tetris\ in~\cite{AboKhamis:2015} was for data complexity. This implies they could afford to use exponential (in the size of the join query) time algorithm to find an appropriate ordering in which to explore different variables. In \SharpSAT\ instances, we can no longer assume that the number of variables is a constant, and hence we cannot obtain an optimal ordering using a brute force algorithm. We deal with this by designing heuristics that take the structure of \Tetris\ into account.
\item As mentioned earlier, \Tetris\ (like a \DPLL\ procedure) performs a sequence of resolutions and theoretically, it can store the outcomes of all the resolutions it performs. However, for practical efficiency, we use a heuristic to decide which resolution results to cache and which ones to discard. 
\end{enumerate}

Our experimental results are promising. On some natural \SharpSAT\ benchmarks based on counting number of occurrences of small subgraphs in a large graph (which we created), our implementation of \Tetris\ is at least two orders of magnitude (and in most cases more than three orders of magnitude) faster than the standard model counters (\SharpSATs~\cite{sharpsat}, \Cachet~\cite{sang2004combining} and \dSharp~\cite{muise2012dsharp}). We also compared \Tetris\ with these model counters on standard \SAT\ benchmarks, where \Tetris\ was either comparable or at most 25x slower.

\paragraph{Theoretical Implications} While this paper deals with an experimental validation of the theoretical result from~\cite{AboKhamis:2015}, we believe that it highlights certain theoretical questions that are worth investigating by the database community. We highlight some of our favorite ones that correspond to each of our three main contributions:
\begin{enumerate}
\item {\em Extending \Tetris\ beyond join queries.} As our work has shown, \Tetris\ can be used to solve problem beyond the original natural join computation. Recently, the worst-case optimal join algorithms were shown to be powerful enough to solve problems in host of other areas such as CSPs (of which \MaxSAT\ is a prominent example), probabilistic graphical models and logic~\cite{faq}. (Also see the followup work~\cite{ajar}.) The beyond worst-case results in~\cite{AboKhamis:2015} have so far seemed more of a theoretical novelty. However, given that this paper demonstrates the viability of \Tetris\ in practice, this work opens up the tantalizing possibility of extending the theoretical results of \Tetris\ to problems captured by~\cite{faq,ajar}. Such a result even for \MaxSAT\ would be of interest in practice.

\item {\em Computing orderings efficiently.} \label{implications:ordering} As mentioned earlier, the theoretical results for \Tetris\ assumes that the required ordering among variables can be computed in exponential time. However, for applications in \SAT\ (as well as other areas such as probabilistic graphical models, assuming the question in the item above can be answered), we need to compute orderings that are approximately good in polynomial time. Thus, a further avenue of theoretical investigation is to come up with a polynomial time algorithm to compute the ordering, and to prove some guarantees on the loss of performance from the case where \Tetris\ has access to the `optimal' ordering. Some of the heuristics developed in our paper might prove to be good starting points for this investigation. 
We would like to point that that the importance of efficiently computing variable orderings has been studied a lot in AI and database literature. Some of the very recent work on Generalized Hypertree Decompositions (which are well known to be equivalent to variable elimination orderings) could potentially be useful towards this goal~\cite{gottlob}.

\item {\em Time-space tradeoff.} \label{implications:time-space} Recent results on worst-case optimal algorithms to compute natural joins~\cite{NPRR,V14,NRR13} and to compute joins with functional dependencies~\cite{ANS16} all focus exclusively on time complexity. However, as highlighted by our work, being more prudent with space usage in fact benefits actual performance. This point was also indirectly highlighted in~\cite{AboKhamis:2015}, where it was shown that resolution schemes that did not cache their intermediate results are strictly less powerful than those that do (in the context of computing the natural join). However, we believe that a systematic theoretical study of the tradeoff between time and space needed to compute the natural join is an attractive route to pursue.
\end{enumerate}

   We will begin in Section \ref{sec:bg} by introducing the fundamental concepts necessary to understand both \SAT\ problems and details on \Tetris\ itself, all while giving a hands-on example of how \Tetris\ would handle a toy example. From there, we will move into Section \ref{sec:cw}, an in-depth analysis of our major contributions. Afterwards, we will continue with our experimental results in Section \ref{sec:res}. Then we will discuss related work in the field in Section \ref{sec:rw}. 

\section{Background}
\label{sec:bg}
\label{SEC:BG}
In this section, we will introduce the concepts necessary to understand how \Tetris\ functions, introduce the concept of resolutions, and walk through how \Tetris\ would handle a simple input. \par

\subsection{\SAT\ and Boxes}
We begin by defining several key terms and ideas. Recall that a \SAT\ problem consists of a series of Boolean variables, $x_1, x_2, ..., x_n$, joined together in a series of \sfAND\ and \sfOR\ clauses. Problems are generally presented in the conjunctive normal form (CNF), a simplification wherein the entire formula is written as a series of \sfANDs over a set of disjunctive clauses. One such example would be $(x_1 \vee x_2) \wedge (x_1 \vee \bar{x_2})$. A solution, or satisfying assignment, to a \SAT\ problem is an assignment of \true\ or \false\ to each of the variables such that the Boolean formula is satisfied; that is, that all clauses are satisfied. \par
    Next, we will consider the idea of boxes, which is how our algorithm will interpret \SAT\ problems. Each box is an $n$-dimensional structure in $\{0, 1\}^{n}$, where $n$ is the number of variables in the original \SAT\ problem. We will define this set as the {\em output space}; that is, all potential outputs will be elements of this set. Each of our boxes exists within this hypercube, and along each dimension has the value 0, the value 1, or extends along the full length of the edge. The reason for this is simple: 0 corresponds to \false, 1 to \true, and the length of the edge to both. Henceforth, we will use $\lambda$ to refer to edges with length 1. We thus form the following definition: \par

\begin{defn}[Box Notation]
A box takes the form $\langle b_1, b_2, ... , b_n\rangle$, where each $b_i \in \{ T, F, \lambda \} $.
\end{defn}

 Observe that, from these definitions, we can consider every assignment to be a 0-dimensional box; this will be important later. \par

    Then, our goal will be to find the set of points within the output space that are not contained (see Definition \ref{def:contains}) by any boxes. Any such point will be termed an {\em output point}, and the goal of an algorithm working on these boxes is to find all such output points. 

\begin{defn}[Containment]
\label{def:contains}
A box $b$ is said to contain another box $c$ if, for all points $p \in \{0, 1\}^n$ such that $p \in c$, it is true that $p \in b$. Equivalently, the box $\langle b_1, b_2, ... , b_n \rangle$ contains the box $\langle c_1, c_2, ... , c_n \rangle$ if, for all $i$, $b_i = c_i$ or $b_i = \lambda$.
\end{defn}

    However, there is one key difference between these two representations. Each clause in a \CNF\ formula is essentially a subproblem wherein at least one variable's assignment must match its value in the clause for the assignment to possibly be satisfying. But with boxes, the exact opposite is true: if an assignment matches the value for the boxes on all non-$\lambda$ dimensions, we reject the assignment. In other words, if we consider a geometric visualization of these boxes, any and all assignments that fall within a box are rejected. \par
    Hence, our next step is to devise a means by which to convert any given \SAT\ problem in \CNF\ form to the boxes format that \Tetris\ can understand. As follows from our above observation, the most important step is simply the negation of the \CNF\ formula; the rest is all bookkeeping. For the exact algorithm, see Algorithm \ref{algo:convert}. \par

\begin{algorithm}[H]
\caption{Conversion from CNF to boxes}
\label{algo:convert}
\begin{algorithmic}[1]
\FOR{each CNF clause}
\STATE {Negate the clause}
\STATE {Set all $\bar{x_i}$ to F, and all $x_j$ to T}
\STATE {Set all variables not present in the clause to $\lambda$.}
\STATE {Insert into the database} \COMMENT{\texttt{See Definition \ref{database:def}}}
\ENDFOR
\end{algorithmic}
\end{algorithm}

Let us consider the following toy example CNF problem:
\begin{example}
$(x_1 \vee x_2) \wedge (x_1 \vee \bar{x_2}) \wedge (x_2 \vee x_3)$. 
\label{exprob1}
\end{example}
Our first step is to negate each clause, which will give us a disjunctive normal form (DNF) co-problem: $(\bar{x_1} \wedge \bar{x_2}) \vee (\bar{x_1} \wedge x_2) \vee (\bar{x_2} \wedge \bar{x_3})$.

Next, we will convert to boxes (by replacing a variable with $T$ and its negation with $F$) and add all missing variables (as $\lambda$). After conversion, our three clauses become $\langle F,F,\lambda \rangle$, $\langle F,T, \lambda \rangle$, and $\langle \lambda, F,F\rangle$. \par

\begin{figure}[h]
\begin{center}
\begin{tikzpicture}

\filldraw[color=black, fill=red] (-3,0) -- (-3.707, -.707) -- (-3.507, -.707) -- (-2.8, 0) -- cycle;
\draw [ultra thick] (-3, 0) -- (-3,1.0) node (yaxis) [above] {$x_2$};
\draw[ultra thick] (-3, 0) -- (-2.0,0) node (xaxis) [right] {$x_1$};
\draw[ultra thick] (-3, 0) -- (-3.707, -.707) node (zaxis) [below] {$x_3$};
\node[] at (-2.0, -.3) {1};
\node[] at (-3.3, 1.0) {1};
\node[] at (-3.907, -.607) {1};
\node[] at (-3.2, .1) {0};

\node[] at (-2.8, -2) {$(x_1 \vee x_2), \langle F, F, \lambda \rangle$};

\filldraw[color=black, fill=red] (0,1) -- (-.707, .293) -- (-.507, .293) -- (.2, 1) -- cycle;
\draw [ultra thick] (0, 0) -- (0,1.0) node (yaxis2) [above] {$x_2$};
\draw[ultra thick] (0, 0) -- (1.0,0) node (xaxis2) [right] {$x_1$};
\draw[ultra thick] (0, 0) -- (-.707, -.707) node (zaxis2) [below] {$x_3$};
\node[] at (1.0, -.3) {1};
\node[] at (-.3, 1.0) {1};
\node[] at (-.907, -.607) {1};
\node[] at (-.2, .1) {0};

\node[] at (.2, -2) {$(x_1 \vee \bar{x_2}), \langle F, T, \lambda \rangle$};

\filldraw[color=black, fill=red] (3, .1) -- (4, .1) -- (4, 0) -- (3,0) -- cycle;
\draw [ultra thick] (3, 0) -- (3,1) node (yaxis3) [above] {$x_2$};
\draw[ultra thick] (3, 0) -- (4,0) node (xaxis3) [right] {$x_1$};
\draw[ultra thick] (3, 0) -- (2.293, -.707) node (zaxis3) [below] {$x_3$};
\node[] at (4.0, -.3) {1};
\node[] at (2.7, 1.0) {1};
\node[] at (2.093, -.607) {1};
\node[] at (2.8, .1) {0};

\node[] at (3.2, -2) {$(x_2 \vee x_3), \langle \lambda, F, F \rangle$};

\end{tikzpicture}
\caption{Our starting boxes, and the corresponding \SAT\ clauses. While boxes are, technically speaking, strictly the corners of what we depict as the boxes, we depict them with the edges and surfaces drawn for the purpose of visual clarity.}
\end{center}
\end{figure}
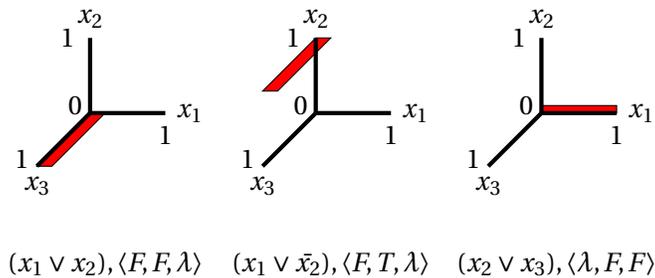

At this point, it is time to insert the boxes into our data structure. Let us list the fundamental operations the data structure must be able to perform:
\begin{defn}[\Tetris\ data structure]
\label{database:def}
The \Tetris\ data structure shall be able to perform the following operations:\\
1) \Insert, input is the box to be inserted, no output\\
2) \Contains, input is the box we are seeing if the structure contains, output is the containing box (see Definition \ref{def:contains})\\
3) \GetAllContainingBoxes, input is the box we are seeing if the structure contains, output is the set of all containing boxes\\
\end{defn}

We will return to the details of the data structure implementation in Section~\ref{sec:ds}.

\subsection{Resolution}
\label{SEC:RESOLUTION}
We then come to the concept of resolution, a key aspect of \Tetris\ and most state-of-the-art \SAT\ solvers. Resolution can be defined over both \CNF\ clauses and boxes; let us begin with the former. Let us consider two clauses in our \CNF\ Example \ref{exprob1} once again: specifically, $(x_1 \vee x_2)$ and $(x_1 \vee \bar{x_2})$. We see that these are two very similar clauses; they differ only in that the $x_2$ term is negated in one and not the other. Therefore, we can {\em resolve} these two clauses by removing the $x_2$ term and then taking the \sfOR\ of all remaining variables. In this case, this gives us $(x_1)$. We then remove the original two clauses from the \CNF\ problem and insert this new clause in its place. This is a significant simplification. \par
Similarly, we can resolve any two clauses such that there is exactly one pivot point, by which we mean a variable that appears in both clauses, but is negated in one and not the other. For instance, looking back at our example, we can also resolve $(x_1 \vee \bar{x_2})$ with $(x_2 \vee x_3)$ to form $(x_1 \vee x_3)$ In this case, we would not be able to remove the original two clauses, but we would have gained information. 
 Let us now formally define this process:
\begin{defn}[Resolution on Clauses]
Two clauses $(x_{i_1} \vee x_{i_2} \vee... \vee x_{i_m} \vee v)$ and $(y_{j_1} \vee y_{j_2} \vee ... \vee y_{j_l} \vee \bar{v}) $, $i \in I$, $j \in J$, $I, J \subset [n]$, can be resolved if and only if there exists exactly one variable $v$, the pivot point, such that $v \in x$ and $\bar{v} \in y$.\\
The resolution of the two clauses is $(x_{i_1} \vee x_{i_2} \vee ... \vee x_{i_m} \vee y_{j_1} \vee y_{j_2} \vee ... \vee y_{j_m})$.
\end{defn}
\par
Since boxes are simply another representation of the same problem, it follows that resolution can be performed on boxes as well.  First, we will require that there must exist exactly one variable on which one box is true and the other box is false.\footnote{This is exactly analogous to the requirement that we resolve on a pivot point in the clause version.} Call this the {\em pivot variable}. In the output, set this variable to $\lambda$. Then, for each other variable, if it is $T$  in one or both boxes, set it to $T$ in the output box; if it is $F$ in one or both boxes, set it to $F$ in the output box; and if both variables are $\lambda$, then the resolution of the two is also $\lambda$. \par
    We see two possible resolutions in our Example \ref{exprob1}. The resolution of $\langle F,F,\lambda\rangle$ and $\langle F,T,\lambda\rangle$, two co-planar and parallel edges, is the square $\langle F,\lambda, \lambda\rangle$, as depicted in Figure \ref{resolution1}, and the resolution of the askew edges $\langle F,T,\lambda \rangle$ and $\langle\lambda ,F,F\rangle$ is the edge $\langle F,\lambda ,F\rangle$, as depicted in Figure \ref{resolution3}. For a formal definition of the resolution operator, henceforth $\oplus$, see Definition \ref{def:boxres}. \par

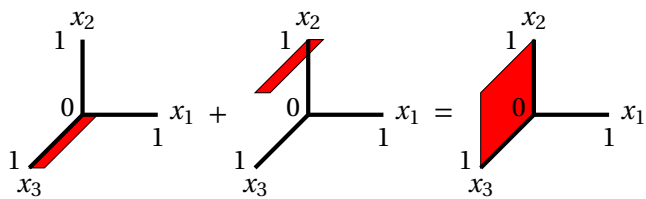
\begin{figure}[H]
\begin{center}

\begin{tikzpicture}

\filldraw[color=black, fill=red] (-3,0) -- (-3.707, -.707) -- (-3.507, -.707) -- (-2.8, 0) -- cycle;
\draw [ultra thick] (-3, 0) -- (-3,1.0) node (yaxis) [above] {$x_2$};
\draw[ultra thick] (-3, 0) -- (-2.0,0) node (xaxis) [right] {$x_1$};
\draw[ultra thick] (-3, 0) -- (-3.707, -.707) node (zaxis) [below] {$x_3$};
\node[] at (-2.0, -.3) {1};
\node[] at (-3.3, 1.0) {1};
\node[] at (-3.907, -.607) {1};
\node[] at (-3.2, .1) {0};

\node[] at (-1.2, 0) {\large \em $+$};

\filldraw[color=black, fill=red] (0,1) -- (-.707, .293) -- (-.507, .293) -- (.2, 1) -- cycle;
\draw [ultra thick] (0, 0) -- (0,1.0) node (yaxis2) [above] {$x_2$};
\draw[ultra thick] (0, 0) -- (1.0,0) node (xaxis2) [right] {$x_1$};
\draw[ultra thick] (0, 0) -- (-.707, -.707) node (zaxis2) [below] {$x_3$};
\node[] at (1.0, -.3) {1};
\node[] at (-.3, 1.0) {1};
\node[] at (-.907, -.607) {1};
\node[] at (-.2, .1) {0};

\node[] at (1.8, 0) {\large \em $=$};

\filldraw[color=black, fill=red] (3, 0) -- (3,1) -- (2.293, .293) -- (2.293, -.707) -- cycle;
\draw [ultra thick] (3, 0) -- (3,1) node (yaxis3) [above] {$x_2$};
\draw[ultra thick] (3, 0) -- (4,0) node (xaxis3) [right] {$x_1$};
\draw[ultra thick] (3, 0) -- (2.293, -.707) node (zaxis3) [below] {$x_3$};
\node[] at (4.0, -.3) {1};
\node[] at (2.7, 1.0) {1};
\node[] at (2.093, -.607) {1};
\node[] at (2.8, .1) {0};

\end{tikzpicture}
\caption{The resolution of $\langle F,F,\lambda \rangle$ and $\langle F,T,\lambda\rangle$ on the vertex $x_2$ is the square $\langle F,\lambda , \lambda \rangle$. This is equivalent to $ (x_1 \vee x_2)$ resolved with $(x_1 \vee \bar{x_2})$  being the clause $(x_1)$.}
\label{resolution1}
\end{center}
\end{figure}

\begin{figure}[H]
\begin{center}
\begin{tikzpicture}

\filldraw[color=black, fill=red] (-4,1) -- (-4.707, .293) -- (-4.507, .293) -- (-3.8, 1) -- cycle;
\draw [ultra thick] (-4, 0) -- (-4,1.0) node (yaxis) [above] {$x_2$};
\draw[ultra thick] (-4, 0) -- (-3.0,0) node (xaxis) [right] {$x_1$};
\draw[ultra thick] (-4, 0) -- (-4.707, -.707) node (zaxis) [below] {$x_3$};
\node[] at (-3.0, -.3) {1};
\node[] at (-4.3, 1.0) {1};
\node[] at (-4.907, -.607) {1};
\node[] at (-4.2, .1) {0};

\node[] at (-1.7, 0) {\large \em $+$};

\filldraw[color=black, fill=red] (0, .1) -- (1, .1) -- (1, 0) -- (0,0) -- cycle;
\draw [ultra thick] (0, 0) -- (0,1.0) node (yaxis2) [above] {$x_2$};
\draw[ultra thick] (0, 0) -- (1.0,0) node (xaxis2) [right] {$x_1$};
\draw[ultra thick] (0, 0) -- (-.707, -.707) node (zaxis2) [below] {$x_3$};
\node[] at (1.0, -.3) {1};
\node[] at (-.3, 1.0) {1};
\node[] at (-.907, -.607) {1};
\node[] at (-.2, .1) {0};

\node[] at (2.3, 0) {\large \em $=$};

\filldraw[color=black, fill=red] (4.01, 0) -- (4.01,1) -- (4.1, 1) -- (4.1, 0) -- cycle;
\draw [ultra thick] (4, 0) -- (4,1) node (yaxis3) [above] {$x_2$};
\draw[ultra thick] (4, 0) -- (5,0) node (xaxis3) [right] {$x_1$};
\draw[ultra thick] (4, 0) -- (3.293, -.707) node (zaxis3) [below] {$x_3$};
\node[] at (5.0, -.3) {1};
\node[] at (3.7, 1.0) {1};
\node[] at (3.093, -.607) {1};
\node[] at (3.8, .1) {0};

\end{tikzpicture}
\caption{The resolution of $\langle F,T,\lambda \rangle$ and $\langle\lambda ,F,F\rangle$ on the vertex $x_2$ is the edge $\langle F,\lambda ,F\rangle$ This is equivalent to $(x_1 \vee \bar{x_2})$ resolved with $(x_2 \vee x_3)$ being the clause $(x_1 \vee x_3)$.}
\label{resolution3}

\end{center}
\end{figure}
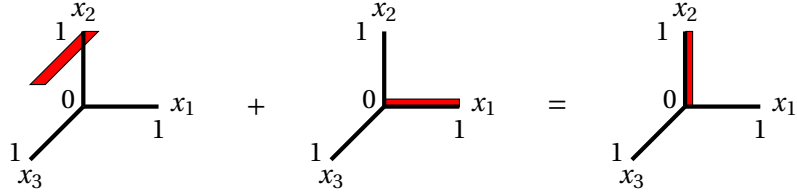

\begin{defn}[Resolution on Boxes]
\label{def:boxres}
Two boxes $\langle b_1, b_2, ..., b_n \rangle$ and $\langle c_1, c_2, ..., c_n \rangle$ can be resolved  if and only if there exists exactly one $i$ such that $b_i$ is \true\ and $c_i$ is \false, or vice-versa. \\
In the resolved box $a$, $a_i = \lambda$. Each $a_j, j \ne i$ is equal to $b_j \oplus c_j$, where $\oplus$ is defined as follows:\\
$T \oplus T = T$\\
$F \oplus F = F$\\
$\lambda \oplus T = T$\\
$\lambda \oplus F = F$\\
$T \oplus \lambda = T$\\
$F \oplus \lambda = F$\\
$\lambda \oplus \lambda = \lambda$\\
$T \oplus F$ is undefined.\\
\end{defn}

   Observe that resolution on boxes and resolution on clauses are identical:

\begin{lemma}
\label{lemma}
\label{app:lemma}
Resolution on boxes, with the additional restriction that exactly one variable must be true  in one box and false in the other, is exactly equivalent to resolution on SAT clauses.
\end{lemma} 
\begin{proof}
Let $(x_{i_1} \vee ... \vee x_{i_m})$ and $(y_{j_1} \vee ... \vee y_{j_l}), i \in I, j \in J, where I and J \subset [n],$ be the clauses we are resolving. Assume WLOG that $i_1 = j_1$ is the pivot point. Then the resolved clause is $(x_{i_2} \vee ... \vee x_{i_m} \vee y_{j_2} \vee ... \vee y_{j_l})$. \\
The boxes equivalent to our starting to clauses are $\langle b_1, ... , b_n \rangle $ and $\langle c_1, ..., c_n \rangle $, where $b_k = F $ if $x_k \in x, b_k = T$ if $\bar{x_k} \in x,$ and $\lambda$ otherwise, and $c_k$ is defined similarly with respect to $y$. The resolution of these boxes $a$ is then defined as $\langle a_1, ..., a_n \rangle$, where $a_k = \lambda$ if $k = i_1 = j_1$, and $a_k = b_k \oplus c_k$ otherwise. \\
Now, let us calculate the box equivalent of the output of the clause-based resolution, $t$. It can be shown that $t_k = \lambda$ if $k = i_1 = j_1$, $t_k = x_k = y_k$ if $k \in I$, $k \in J$ and $k \ne i_1$, $t_k = x_k$ if $k \in I$ and $k \notin J$, $t_k = y_k$ if $k \in J$ and $k \notin I$, and $t_k = \lambda$ if $k \notin I $ and $k \notin J$. Inspection with the above definition of $\oplus$ reveals that $t$ is exactly equivalent to $a$. Since the same problem with arbitrary, equivalent inputs produced equivalent outputs, the two operations must be equivalent.
\end{proof}

\Tetris\ introduces one additional restriction on resolution. 
\begin{defn}[Resolution on Boxes in \Tetris]
Two boxes $b$ and $c$ can be resolved if and only if there is exactly one spot $i$ such that $b_i$ = T and $c_i$ = F, or vice-versa, and for all $j > i$, $b_j = c_j = \lambda$.
\end{defn}
In other words, we will demand that that the pivot variable be the final non-$\lambda$ variable. Therefore, while \Tetris\ will perform the resolution of $\langle F,F,\lambda\rangle$ and $\langle F,T,\lambda \rangle$ (Figure \ref{resolution1}), it will not perform the resolution of $\langle F,T,\lambda \rangle$ and $\langle\lambda ,F,F\rangle$ (Figure \ref{resolution3}). We see, then, that the ordering of the variables determines whether or not a resolution is even possible. This makes determining the global ordering of the variables a key issue, as mentioned earlier as Theoretical Implication \ref{implications:ordering}, which we will address later in Section \ref{sec:gvo}. \par
    In general, \Tetris\ performs resolution on pairs of recently found boxes. Let $k$ be the location of the last non-$\lambda$ variable in a box $b$. Then $b_k$ must be either \true\ or \false. If it is \false, we will store the box for future use. If it is \true, then we will take this box $b$ and resolve it with the stored box with the same value for $k$ whose last non-$\lambda$ variable was \false. By doing so, we will guarantee the production of a box where the last $n-k+1$ variables have the value $\lambda$. For more details on how this works, along with the reasoning for why such pairs can always be found, see Section \ref{Tetris}.

\subsection{\Tetris}
\label{Tetris}
For now, let us return to our Example \ref{exprob1}. When we last left off, we were just inserting the three clauses into our data structure, which was loosely defined in Definition \ref{database:def}. For a formal definition of the database and details for how it allows the set of boxes \Tetris\ knows about to be quickly and efficiently queried, see Section \ref{sec:ds}; for now, one can simply assume it to be a trie-based structure. Additionally, we can resolve the first two boxes while leaving the third untouched, as in Figure \ref{resolution1}; therefore, the database will contain exactly the boxes $\langle F,\lambda ,\lambda \rangle$ and $\langle \lambda ,F,F\rangle$ (see Figure \ref{dbstate1}).  Furthermore, we will prepare an empty array of boxes $L$ of size $n$, which will be used later. The purpose of this array is to store and retrieve boxes that we wish to resolve with other boxes. \par
    Now that we have our database established, it is time to perform \Tetris\ proper. The basic idea here is very simple. We will pick a point $P$ in the output space, which we will call the {\em probe point}; recall that this point is itself a 0-dimensional box. We then determine whether or not any box in the database contains this point. If one does, we will store this box in an additional data structure referred to as the {\em cache}, which functions identically to the main database, and probe a new point. If no box contains $P$, we will list the point as a solution and furthermore add this point into the cache. Along the way, we will perform resolution in order to create new and larger boxes. This process continues until the entirety of the output space is covered by a single box, at which point we must have found every output point and are done. Algorithm \ref{algo:genTetris} has the details. It should be noted that this algorithm was originally presented recursively in \cite{AboKhamis:2015}; here, we present it iteratively both for the purposes of speed and because this allows for non-chronological backtracking; in other words, we can backtrack more than one layer at a time.\par

\begin{algorithm}[h]
\caption{Advance(box $b$, probe point $\&p$)  \texttt{     (note: $p$ is a global variable)}}
\label{algo:adv}
\begin{algorithmic}[1]
\WHILE{$b$ \Contains\ $p$}
\IF{the last non-$\lambda$ variable of $p$ is $F$}
\STATE{Set that variable to $T$} \COMMENT{\texttt{Return to the previous branching point and take the right, or \true, branch}}
\ELSE 
\WHILE{the last non-$\lambda$ variable of $p$ is $T$}
\STATE{Set that variable to $\lambda$} \COMMENT{\texttt{Return to the most recent level where we branched left}}
\ENDWHILE
\STATE{Set the last non-$\lambda$ variable of $p$ to $T$} \COMMENT{\texttt{Branch right here}}
\STATE{Replace all $\lambda$s after this variable with $F$} \COMMENT{\texttt{Repeatedly branch left}}
\ENDIF
\ENDWHILE
\end{algorithmic}
\end{algorithm}

\begin{algorithm}[h]
\caption{General \Tetris\ for \SAT}
\label{algo:genTetris}
\begin{algorithmic}[1]
\STATE {Establish variable ordering}
\STATE {Build the database $D$ using Algorithm \ref{algo:convert}}
\STATE {$C \gets \emptyset$}
\STATE {$L \gets $  An empty array of size $n$} \COMMENT{\texttt{This array is implicit in \cite{AboKhamis:2015}}}
\STATE {$p \gets \langle F, F, ..., F\rangle$ }
\WHILE{$\langle \lambda, \lambda, ..., \lambda \rangle \notin C$}
  \IF {($b \gets  C$.\Contains ($p$)) is nonempty}
     \STATE {Advance($b,p)$} \COMMENT{\texttt{Advance (see Algorithm \ref{algo:adv}) the probe point past $b$}}
  \ELSIF {($A \gets  D$.\GetAllContainingBoxes ($p$)) is nonempty}
     \FORALL {boxes $b \in A$}
     \STATE {$C.\Insert(b)$} 
     \STATE {Advance($b,p)$} \COMMENT{\texttt{Advance the probe point past $b$}}
     \ENDFOR
  \ELSE 
     \STATE {Add $p$ to the output}\COMMENT{\texttt{There is no containing box, so $p$ is an output point}}
     \STATE {$C. \Insert (p)$}
     \STATE {Advance($p,p)$} \COMMENT{\texttt{Advance the probe point past itself}}
  \ENDIF
  \STATE {$k \gets$ the location of the last non-$\lambda$ variable in $b$ w.r.t. the variable ordering}
  \IF {$b_k = F$}
      \STATE {$L[k] = b$} \COMMENT{\texttt{Store the most recent left-branching box for a given depth}}
  \ELSE     
      \STATE {$r \gets b \oplus L[k]$}                  \COMMENT{\texttt{Resolve this right-branching box with the corresponding left-branching box}}
      \STATE {$C. \Insert (r)$}
  \ENDIF
\ENDWHILE
\end{algorithmic}
\end{algorithm}

Now, let us consider how this algorithm behaves with regards to our earlier Example \ref{exprob1}. We pick as our first probe point $\langle F,F,F\rangle$, giving us the situation illustrated in Figure \ref{dbstate1}. We first scan our local cache, $C$, for any boxes that contain this point; however, since this is the first probe point, the cache is trivially empty. Next, we scan the database $D$. $D$ contains all the boxes corresponding to the clauses in the original \SAT\ problem. The database just so happens to contain two containing boxes; for reasons that will become clear shortly, the operation will choose to output $\langle F,\lambda, \lambda\rangle$.  We insert the box $\langle F, \lambda, \lambda \rangle$ into $C$. \par
Our next task is to advance the probe point until it lies beyond our box. To do this, we proceed according to Algorithm \ref{algo:adv}. The idea is to think of the set of all possible probe points as a tree that we are performing a depth-first search on, with $F$ representing left-branching paths and $T$ representing right-branching paths. We will continue along this depth-first search until we find a point not covered by the most recently discovered box. This takes us to $\langle T, F, F\rangle$. Note that if the database had fetched $\langle F, F, \lambda \rangle$, we would not have been able to advance the probe point as far. \par
   Finally, we insert our containing box into the array $L$ at location 1, since only the first variable is non-$\lambda$, for future use.  We know to do insertion here, rather than trying to resolve with a non-existent box, because the value of that first non-$\lambda$ variable is $\false$. This takes us to the situation depicted in Figure \ref{dbstate2}. \par

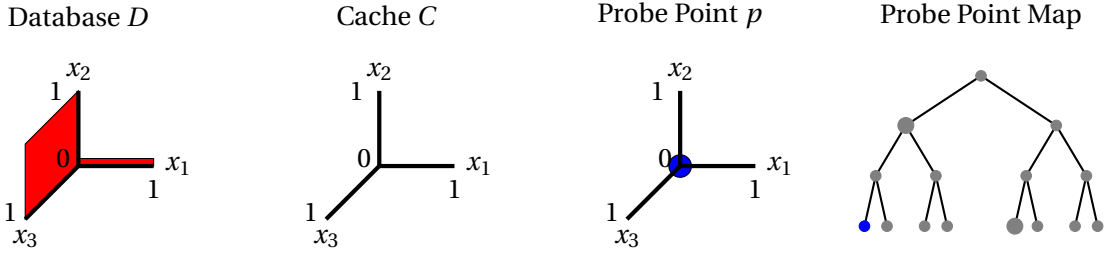
\begin{figure}[h]
\begin{center}
\begin{tikzpicture}

\node[] at (-4.0, 2) {Database $D$};

\filldraw[color=black, fill=red] (-4,0) -- (-4.707, -.707) -- (-4.707, .293) -- (-4, 1) -- cycle;
\filldraw[color=black, fill=red] (-4, .1) -- (-3, .1) -- (-3, 0) -- (-4,0) -- cycle;
\draw [ultra thick] (-4, 0) -- (-4,1.0) node (yaxis) [above] {$x_2$};
\draw[ultra thick] (-4, 0) -- (-3.0,0) node (xaxis) [right] {$x_1$};
\draw[ultra thick] (-4, 0) -- (-4.707, -.707) node (zaxis) [below] {$x_3$};
\node[] at (-3.0, -.3) {1};
\node[] at (-4.3, 1.0) {1};
\node[] at (-4.907, -.607) {1};
\node[] at (-4.2, .1) {0};

\node[] at (.1, 2) {Cache $C$};

\draw [ultra thick] (0, 0) -- (0,1.0) node (yaxis2) [above] {$x_2$};
\draw[ultra thick] (0, 0) -- (1.0,0) node (xaxis2) [right] {$x_1$};
\draw[ultra thick] (0, 0) -- (-.707, -.707) node (zaxis2) [below] {$x_3$};
\node[] at (1.0, -.3) {1};
\node[] at (-.3, 1.0) {1};
\node[] at (-.907, -.607) {1};
\node[] at (-.2, .1) {0};

\node[] at (4.0, 2) {Probe Point $p$};

\filldraw[color=black, fill=blue] (4.00, 0) circle (.15);
\draw [ultra thick] (4, 0) -- (4,1) node (yaxis3) [above] {$x_2$};
\draw[ultra thick] (4, 0) -- (5,0) node (xaxis3) [right] {$x_1$};
\draw[ultra thick] (4, 0) -- (3.293, -.707) node (zaxis3) [below] {$x_3$};
\node[] at (5.0, -.3) {1};
\node[] at (3.7, 1.0) {1};
\node[] at (3.093, -.607) {1};
\node[] at (3.8, .1) {0};

\node[] at (8.0, 2) {Probe Point Map};

\draw[thick] (8,1.2) -- (7, .54) -- (6.6, -.13) -- (6.45, -0.8);
\draw[thick] (6.6, -0.13) -- (6.75, -0.8);
\draw[thick] (7, .54) -- (7.4, -.13) -- (7.25, -.8);
\draw[thick] (7.4, -.13) -- (7.55, -.8);
\draw[thick] (8, 1.2) -- (9, .54) -- (8.6, -.13) -- (8.45, -.8);
\draw[thick] (8.6, -.13) -- (8.75, -.8);
\draw[thick] (9, .54) -- (9.4, -.13) -- (9.55, -.8);
\draw[thick] (9.4, -.13) -- (9.25, -.8);

\filldraw[color=gray] (8,1.2) circle (2pt);
\filldraw[color=gray] (7,.54) circle (3pt);
\filldraw[color=gray] (9,.54) circle (2pt);
\filldraw[color=gray] (6.6,-.13) circle (2pt);
\filldraw[color=gray] (7.4,-.13) circle (2pt);
\filldraw[color=gray] (8.6,-.13) circle (2pt);
\filldraw[color=gray] (9.4,-.13) circle (2pt);
\filldraw[color=blue] (6.45,-.8) circle (2pt);
\filldraw[color=gray] (6.75,-.8) circle (2pt);
\filldraw[color=gray] (7.25,-.8) circle (2pt);
\filldraw[color=gray] (7.55, -.8) circle (2pt);
\filldraw[color=gray] (9.55,-.8) circle (2pt);
\filldraw[color=gray] (9.25,-.8) circle (2pt);
\filldraw[color=gray] (8.75,-.8) circle (2pt);
\filldraw[color=gray] (8.45, -.8) circle (3pt);

\end{tikzpicture}
\caption{The initial state of the database $D$, cache $C$, and the location of the first probe point $p$, which will search the output space in a depth-first manner that can be tracked using the map on its right. $D$ is simply the union of all the boxes we created from the initial \SAT\ problem, while $p$ is set to an initial value of $\langle F, F,F\rangle$. $L$ is currently empty.}
\label{dbstate1}
\end{center}
\end{figure}

\begin{figure}[h]
\begin{center}

\begin{tikzpicture}

\node[] at (-4.0, 2) {Database $D$};

\filldraw[color=black, fill=red] (-4,0) -- (-4.707, -.707) -- (-4.707, .293) -- (-4, 1) -- cycle;
\filldraw[color=black, fill=red] (-4, .1) -- (-3, .1) -- (-3, 0) -- (-4,0) -- cycle;
\draw [ultra thick] (-4, 0) -- (-4,1.0) node (yaxis) [above] {$x_2$};
\draw[ultra thick] (-4, 0) -- (-3.0,0) node (xaxis) [right] {$x_1$};
\draw[ultra thick] (-4, 0) -- (-4.707, -.707) node (zaxis) [below] {$x_3$};
\node[] at (-3.0, -.3) {1};
\node[] at (-4.3, 1.0) {1};
\node[] at (-4.907, -.607) {1};
\node[] at (-4.2, .1) {0};

\node[] at (.1, 2) {Cache $C$};

\filldraw[color=black, fill=red] (0,0) -- (-.707, -.707) -- (-.707, .293) -- (0, 1) -- cycle;
\draw [ultra thick] (0, 0) -- (0,1.0) node (yaxis2) [above] {$x_2$};
\draw[ultra thick] (0, 0) -- (1.0,0) node (xaxis2) [right] {$x_1$};
\draw[ultra thick] (0, 0) -- (-.707, -.707) node (zaxis2) [below] {$x_3$};
\node[] at (1.0, -.3) {1};
\node[] at (-.3, 1.0) {1};
\node[] at (-.907, -.607) {1};
\node[] at (-.2, .1) {0};

\node[] at (4.0, 2) {Probe Point $p$};

\filldraw[color=black, fill=blue] (5.00, 0) circle (.15);
\draw [ultra thick] (4, 0) -- (4,1) node (yaxis3) [above] {$x_2$};
\draw[ultra thick] (4, 0) -- (5,0) node (xaxis3) [right] {$x_1$};
\draw[ultra thick] (4, 0) -- (3.293, -.707) node (zaxis3) [below] {$x_3$};
\node[] at (5.0, -.3) {1};
\node[] at (3.7, 1.0) {1};
\node[] at (3.093, -.607) {1};
\node[] at (3.8, .1) {0};

\node[] at (8.0, 2) {Probe Point Map};

\draw[thick] (8,1.2) -- (7, .54) -- (6.6, -.13) -- (6.45, -0.8);
\draw[thick] (6.6, -0.13) -- (6.75, -0.8);
\draw[thick] (7, .54) -- (7.4, -.13) -- (7.25, -.8);
\draw[thick] (7.4, -.13) -- (7.55, -.8);
\draw[thick] (8, 1.2) -- (9, .54) -- (8.6, -.13) -- (8.45, -.8);
\draw[thick] (8.6, -.13) -- (8.75, -.8);
\draw[thick] (9, .54) -- (9.4, -.13) -- (9.55, -.8);
\draw[thick] (9.4, -.13) -- (9.25, -.8);

\filldraw[color=gray] (8,1.2) circle (2pt);
\filldraw[color=orange] (7,.54) circle (3pt);
\filldraw[color=gray] (9,.54) circle (2pt);
\filldraw[color=gray] (6.6,-.13) circle (2pt);
\filldraw[color=gray] (7.4,-.13) circle (2pt);
\filldraw[color=gray] (8.6,-.13) circle (2pt);
\filldraw[color=gray] (9.4,-.13) circle (2pt);
\filldraw[color=gray] (6.45,-.8) circle (2pt);
\filldraw[color=gray] (6.75,-.8) circle (2pt);
\filldraw[color=gray] (7.25,-.8) circle (2pt);
\filldraw[color=gray] (7.55, -.8) circle (2pt);
\filldraw[color=gray] (9.55,-.8) circle (2pt);
\filldraw[color=gray] (9.25,-.8) circle (2pt);
\filldraw[color=gray] (8.75,-.8) circle (2pt);
\filldraw[color=blue] (8.45, -.8) circle (3pt);

\end{tikzpicture}
\caption{The state of the database, cache, and probe point after the first round of the algorithm. Our probe point found the box $\langle F, \lambda, \lambda \rangle$, so it added this box to $C$. Then, $p$ was advanced until it reached a point not contained by this box, which turned out to be $\langle T, F, F\rangle$. $L(1)$ is the box $\langle F, \lambda, \lambda \rangle$, which is the box corresponding to the orange vertex in the map; the array is empty elsewhere.}
\label{dbstate2}
\end{center}
\end{figure}
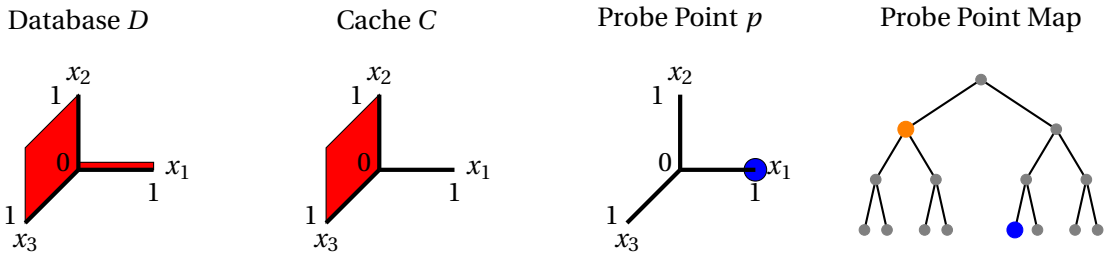

Again we scan $C$, this time for the probe point $\langle T, F, F\rangle$, and again we find no containing box in $C$. So we scan $D$ once again and find the containing box $\langle \lambda ,F,F\rangle$. We then insert this box into the cache and advance the probe point to $\langle T,F,T\rangle$. This time, although our containing box features a $\lambda$ at the first location, we determine the location in $L$ into which we will insert based on the location of the last non-$\lambda$ variable, so we insert it into $L[3]$. \par

    This time, we find no containing boxes in either the cache or the database. Therefore, we have found an output point (see Figure~\ref{dbstate3} for an illustration). We add $\langle T,F,T\rangle$ to our output set, then add the box $\langle T,F,T\rangle$ to our cache, which marks the point as found. \par

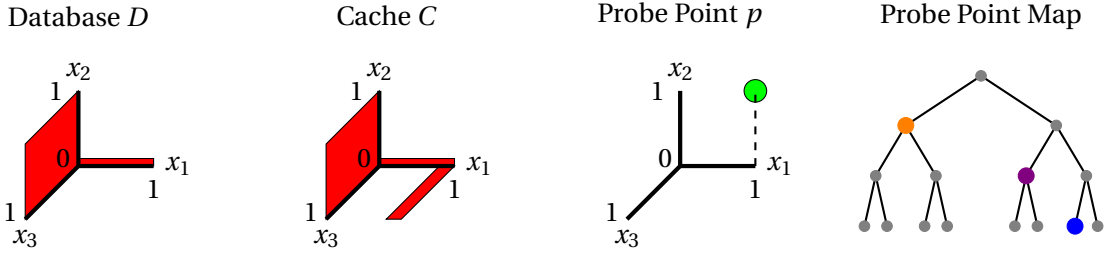
\begin{figure}[h]
\begin{center}

\begin{tikzpicture}

\node[] at (-4.0, 2) {Database $D$};

\filldraw[color=black, fill=red] (-4,0) -- (-4.707, -.707) -- (-4.707, .293) -- (-4, 1) -- cycle;
\filldraw[color=black, fill=red] (-4, .1) -- (-3, .1) -- (-3, 0) -- (-4,0) -- cycle;
\draw [ultra thick] (-4, 0) -- (-4,1.0) node (yaxis) [above] {$x_2$};
\draw[ultra thick] (-4, 0) -- (-3.0,0) node (xaxis) [right] {$x_1$};
\draw[ultra thick] (-4, 0) -- (-4.707, -.707) node (zaxis) [below] {$x_3$};
\node[] at (-3.0, -.3) {1};
\node[] at (-4.3, 1.0) {1};
\node[] at (-4.907, -.607) {1};
\node[] at (-4.2, .1) {0};

\node[] at (.1, 2) {Cache $C$};

\filldraw[color=black, fill=red] (0,0) -- (-.707, -.707) -- (-.707, .293) -- (0, 1) -- cycle;
\filldraw[color=black, fill=red] (0, .1) -- (1, .1) -- (1, 0) -- (0,0) -- cycle;
\filldraw[color=black, fill=red] (.8,0) -- (.093, -.707) -- (.293, -.707) -- (1.0, 0) -- cycle;
\draw [ultra thick] (0, 0) -- (0,1.0) node (yaxis2) [above] {$x_2$};
\draw[ultra thick] (0, 0) -- (1.0,0) node (xaxis2) [right] {$x_1$};
\draw[ultra thick] (0, 0) -- (-.707, -.707) node (zaxis2) [below] {$x_3$};
\node[] at (1.0, -.3) {1};
\node[] at (-.3, 1.0) {1};
\node[] at (-.907, -.607) {1};
\node[] at (-.2, .1) {0};

\node[] at (4.0, 2) {Probe Point $p$};

\filldraw[color=black, fill=green] (5.00, 1.0) circle (.15);
\draw [ultra thick] (4, 0) -- (4,1) node (yaxis3) [above] {$x_2$};
\draw[ultra thick] (4, 0) -- (5,0) node (xaxis3) [right] {$x_1$};
\draw[ultra thick] (4, 0) -- (3.293, -.707) node (zaxis3) [below] {$x_3$};
\draw[thick, dashed] (5, 0) -- (5, 1);
\node[] at (5.0, -.3) {1};
\node[] at (3.7, 1.0) {1};
\node[] at (3.093, -.607) {1};
\node[] at (3.8, .1) {0};

\node[] at (8.0, 2) {Probe Point Map};

\draw[thick] (8,1.2) -- (7, .54) -- (6.6, -.13) -- (6.45, -0.8);
\draw[thick] (6.6, -0.13) -- (6.75, -0.8);
\draw[thick] (7, .54) -- (7.4, -.13) -- (7.25, -.8);
\draw[thick] (7.4, -.13) -- (7.55, -.8);
\draw[thick] (8, 1.2) -- (9, .54) -- (8.6, -.13) -- (8.45, -.8);
\draw[thick] (8.6, -.13) -- (8.75, -.8);
\draw[thick] (9, .54) -- (9.4, -.13) -- (9.55, -.8);
\draw[thick] (9.4, -.13) -- (9.25, -.8);

\filldraw[color=gray] (8,1.2) circle (2pt);
\filldraw[color=orange] (7,.54) circle (3pt);
\filldraw[color=gray] (9,.54) circle (2pt);
\filldraw[color=gray] (6.6,-.13) circle (2pt);
\filldraw[color=gray] (7.4,-.13) circle (2pt);
\filldraw[color=violet] (8.6,-.13) circle (3pt);
\filldraw[color=gray] (9.4,-.13) circle (2pt);
\filldraw[color=gray] (6.45,-.8) circle (2pt);
\filldraw[color=gray] (6.75,-.8) circle (2pt);
\filldraw[color=gray] (7.25,-.8) circle (2pt);
\filldraw[color=gray] (7.55, -.8) circle (2pt);
\filldraw[color=gray] (9.55,-.8) circle (2pt);
\filldraw[color=blue] (9.25,-.8) circle (3pt);
\filldraw[color=gray] (8.75,-.8) circle (2pt);
\filldraw[color=gray] (8.45, -.8) circle (2pt);

\end{tikzpicture}
\caption{The state of the database, cache, and probe point after the first output point is found at $\langle T, F, T\rangle$. In getting here, we first found the box $\langle \lambda, F, F \rangle$, and then probed the output point. After finding the output point, that box was resolved with the aforementioned box to produce $\langle T, F, \lambda\rangle$, which was also added to $C$. Note that, while the box $\langle \lambda, F, \lambda \rangle$ could be produced at this juncture, \Tetris\ will not do so. $L(1)$ contains $\langle F, \lambda, \lambda \rangle$ (the orange dot); $L(2)$ contains $\langle T, F, \lambda \rangle$ (the purple dot); $L(3)$ is empty. }
\label{dbstate3}
\end{center}
\end{figure}

    At this juncture, we find that the last non-$\lambda$ variable is at location 3, but this time, it is \true. Therefore, we will extract the same-length box we stored previously at $L[3]$ and resolve it with this box. We know that this will be a legal resolution because we are scanning the output space in a tree-like fashion. This means that, when retreating from a right branch, the box containing the corresponding left branch must be able to contain the right branch if the final non-$\lambda$ variable were set to $\lambda$ instead. It follows that this final non-$\lambda$ variable must be the one and only pivot point between the two.\par
     Therefore, we can and do perform this resolution; in this example, it is $\langle T,F,T\rangle$ resolved with $\langle \lambda , F, F\rangle$. This outputs the box $\langle T,F,\lambda \rangle$. We furthermore store this box in $L$ at location 2, since this box ends with false at that index. \par

    We continue forth with probe points $\langle T,T,F\rangle$ and $\langle T,T,T\rangle$. Neither will be found in either $D$ or $C$, and are therefore output points. Both again have their final non-$\lambda$ variables at index 3, with  $\langle T,T,F\rangle$ being inserted into $L$ at that index and then $\langle T,T,T\rangle$ recovering that box so it can resolve with it to form the box  $\langle T,T,\lambda \rangle$. This time, the output of our resolution ends with $\true$, so we recover the box at index 2 in $L$, $\langle T,F,\lambda \rangle$, and take the resolution of these two boxes, giving us $\langle T, \lambda, \lambda \rangle$. Once again this ends in T, so we can resolve it with the box we found back at the beginning that has been waiting in slot 1, $\langle F, \lambda, \lambda \rangle$, to form the box $\langle \lambda, \lambda, \lambda \rangle$. This box completely covers the output space; therefore, the algorithm knows that it has found all possible output points and terminates (see Figure~\ref{dbstate4} for an illustration).

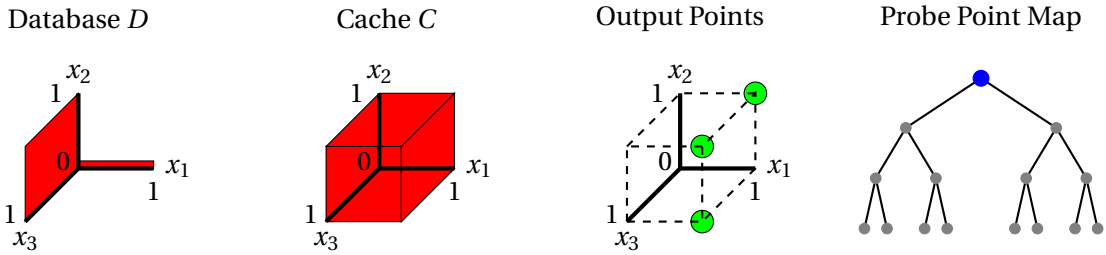
\begin{figure}[h]
\begin{center}

\begin{tikzpicture}
\node[] at (-4.0, 2) {Database $D$};

\filldraw[color=black, fill=red] (-4,0) -- (-4.707, -.707) -- (-4.707, .293) -- (-4, 1) -- cycle;
\filldraw[color=black, fill=red] (-4, .1) -- (-3, .1) -- (-3, 0) -- (-4,0) -- cycle;
\draw [ultra thick] (-4, 0) -- (-4,1.0) node (yaxis) [above] {$x_2$};
\draw[ultra thick] (-4, 0) -- (-3.0,0) node (xaxis) [right] {$x_1$};
\draw[ultra thick] (-4, 0) -- (-4.707, -.707) node (zaxis) [below] {$x_3$};
\node[] at (-3.0, -.3) {1};
\node[] at (-4.3, 1.0) {1};
\node[] at (-4.907, -.607) {1};
\node[] at (-4.2, .1) {0};

\node[] at (.1, 2) {Cache $C$};

\filldraw[color=black, fill=red] (0,0) -- (1,0) -- (0.293,-.707) -- (-.707, -.707) -- cycle;
\filldraw[color=black, fill=red] (0.293,-.707) -- (-.707, -.707) -- (-.707, .293) -- (.293, .293) -- cycle;
\filldraw[color=black, fill=red] (1.0, 1.0) -- (.293, .293) -- (.293, -.707) -- (1.0, 0) -- cycle;
\filldraw[color=black, fill=red] (-.707, .293) -- (.293, .293) -- (1.0, 1.0) -- (0, 1.0) -- cycle;
\draw [ultra thick] (0, 0) -- (0,1.0) node (yaxis2) [above] {$x_2$};
\draw[ultra thick] (0, 0) -- (1.0,0) node (xaxis2) [right] {$x_1$};
\draw[ultra thick] (0, 0) -- (-.707, -.707) node (zaxis2) [below] {$x_3$};
\node[] at (1.0, -.3) {1};
\node[] at (-.3, 1.0) {1};
\node[] at (-.907, -.607) {1};
\node[] at (-.2, .1) {0};

\node[] at (4.0, 2) {Output Points};

\filldraw[color=black, fill=green] (5.00, 1.0) circle (.15);
\filldraw[color=black, fill=green] (4.293, .293) circle (.15);
\filldraw[color=black, fill=green] (4.293, -.707) circle(.15);
\draw [ultra thick] (4, 0) -- (4,1) node (yaxis3) [above] {$x_2$};
\draw[ultra thick] (4, 0) -- (5,0) node (xaxis3) [right] {$x_1$};
\draw[ultra thick] (4, 0) -- (3.293, -.707) node (zaxis3) [below] {$x_3$};
\draw[thick, dashed] (3.293, -.707) -- (4.293, -.707) -- (5.0, 0) -- (5, 1) -- (4.293, .293) -- (4.293, -.707);
\draw[thick, dashed] (3.293, -.707) -- (3.293, .293) -- (4.293, .293);
\draw[thick, dashed] (3.293, .293) -- (4, 1) -- (5, 1);
\node[] at (5.0, -.3) {1};
\node[] at (3.7, 1.0) {1};
\node[] at (3.093, -.607) {1};
\node[] at (3.8, .1) {0};

\node[] at (8.0, 2) {Probe Point Map};

\draw[thick] (8,1.2) -- (7, .54) -- (6.6, -.13) -- (6.45, -0.8);
\draw[thick] (6.6, -0.13) -- (6.75, -0.8);
\draw[thick] (7, .54) -- (7.4, -.13) -- (7.25, -.8);
\draw[thick] (7.4, -.13) -- (7.55, -.8);
\draw[thick] (8, 1.2) -- (9, .54) -- (8.6, -.13) -- (8.45, -.8);
\draw[thick] (8.6, -.13) -- (8.75, -.8);
\draw[thick] (9, .54) -- (9.4, -.13) -- (9.55, -.8);
\draw[thick] (9.4, -.13) -- (9.25, -.8);

\filldraw[color=blue] (8,1.2) circle (3pt);
\filldraw[color=gray] (7,.54) circle (2pt);
\filldraw[color=gray] (9,.54) circle (2pt);
\filldraw[color=gray] (6.6,-.13) circle (2pt);'
\filldraw[color=gray] (7.4,-.13) circle (2pt);
\filldraw[color=gray] (8.6,-.13) circle (2pt);
\filldraw[color=gray] (9.4,-.13) circle (2pt);
\filldraw[color=gray] (6.45,-.8) circle (2pt);
\filldraw[color=gray] (6.75,-.8) circle (2pt);
\filldraw[color=gray] (7.25,-.8) circle (2pt);
\filldraw[color=gray] (7.55, -.8) circle (2pt);
\filldraw[color=gray] (9.55,-.8) circle (2pt);
\filldraw[color=gray] (9.25,-.8) circle (2pt);
\filldraw[color=gray] (8.75,-.8) circle (2pt);
\filldraw[color=gray] (8.45, -.8) circle (2pt);

\end{tikzpicture}

\caption{The state of the database, cache, and probe point after the entire output space has been covered. With each further output point discovered, boxes were added to the cache, which produced a chain of resolutions that eventually resulted in the production of the box $\langle \lambda, \lambda, \lambda \rangle$. Note that there is no longer a probe point, as there is nothing left to probe.}
\label{dbstate4}
\end{center}
\end{figure}

\section{Our Improvements}
\label{sec:cw}
Here, we will discuss the major additions introduced into \Tetris\ in order to handle \CNF\ inputs and increase practical efficiency. These include a new data structure, work on heuristically determining a global variable ordering, and selectively caching only certain boxes.

\subsection{Data Structure and Compression}
\label{Data_Structure}
\label{sec:ds}
\label{SEC:DS}
For its data structure, the original \Tetris\ paper simply states that a trie will suffice to achieve asymptotic runtime guarantees. While this is true, a simple trie still leaves much to be desired; attempts to implement \Tetris\ in such a simple matter produced a system significantly slower than other state-of-the-art model counters. Our contribution is to design a novel system of tries that takes advantage of the nature of the problem space to improve both runtime and memory usage.\par
\subsubsection{Data Structure Description}
    As described above, the database must allow each variable to store three values: \false, \true, and $\lambda$. Therefore, the immediate approach is to use 3-tries as the base data structure. However, when \SharpSAT\ instances routinely have hundreds of variables, this results in an extremely deep problem space that requires a lot of time to probe. Our next step, then, is to compress multiple layers into a single node that can be queried in a single instruction.\par
    To this end, we will first come up with a means to enumerate all possible boxes: 
\begin{defn}
Let $\phi$ be a bijective function from the set of boxes onto the integers.
\end{defn}
\begin{example}
One such way is as follows: Let $\lambda$ be assigned 1; $\false$, 2; $\true$, 3. Then, for each box $\langle b_1, b_2, ..., b_n\rangle$, its numerical value is $3^{n-1}b_1 + 3^{n-2}b_2 + ... + b_n$.\footnote{This is exactly the mapping used in our implementation.} This gives us a bijective ternary numeration. 
\end{example}

From there, we observe that there is exactly one box for which $n$ is 0 (namely, the empty box $\langle \rangle$), three boxes with $n$ equal to 1, nine boxes with $n$ equal to 2, twenty-seven boxes with $n$ equal to 3, and eighty-one boxes with $n$ equal to 4. We will require that a single node within the trie be able to record a box of any of these lengths. Additionally, it must be able to store the children of all possible boxes with $n$ equal to 4. Therefore, by compacting four logical layers into a single layer within the database, the result is a trie that can store 121 possible boxes, the sum of the five aforementioned values, and can have 81 possible children. We will refer to this collection of variables as a cluster. \par
\begin{defn}[Cluster]
A {\em Cluster} is the set of variables that the database handles in a single operation. By default, each cluster contains 4 variables.
\end{defn}

    Of course, this raises a new issue: When checking if the database contains a given input string $x$, there can exist up to sixteen children of $x$ that must be checked (since each $T$ or $F$ can be replaced by a $\lambda$), and an even greater number of boxes that could be contained in this cluster may contain the input string. This creates the need for a way to quickly and efficiently determine if a containing box exists in this cluster and to create the list of children to be searched. \par

\subsubsection{\SIMD -based Trie}
\label{SIMD_sample}
    Here, we take inspiration from EmptyHeaded, a relational database engine \cite{emptyheaded}, and utilize \SIMD. As an example, let us consider the simplified version of the data structure that contains only two layers, and with $n = 3$. Suppose that $\langle F, \lambda, \lambda \rangle$ was known to be a box, and that $\langle \lambda, T, F\rangle$ and $\langle T, T, F \rangle$ are boxes that, to be found, necessitate traversing into child clusters. We can see this depicted in Figure \ref{simd_picture}. Now, we will determine whether or not this data structure contains the box $\langle F, T, F \rangle$. Since each cluster contains two layers, clusters with depth 0 will look at only the first two variables in the input box to determine input; therefore, let us consider the sub-box $\langle F, T \rangle$. Using $\phi$, we can find in a lookup table the two 128-bit bitstrings corresponding to this input sub-box. The first lists the set of boxes that, if they exist within the cluster, would contain $\langle F, T \rangle$, and is the second line of Figure \ref{simd_bitstring}. The second bitstring does much the same for the set of children that, if truncated to two variables, contain $\langle F, T\rangle$. Additionally, a cluster stores two more 128-bit bitstrings: \boxes, which marks the boxes contained by the cluster; and \children, which marks the child nodes of the cluster. The \boxes\ bitstring is specifically the top line of Figure \ref{simd_bitstring}. \boxes\ is marked for the box $\langle F \rangle$ (which is equivalent to $\langle F, \lambda, \lambda \rangle$), while \children\ has bits marked corresponding to the box prefixes $\langle F, T \rangle$ and $\langle T, T \rangle$. \par
    It follows that the intersection of these two pairs of bitstrings is the set of boxes and child prefixes present in the data structure that contain the input; a single AND operation suffices to calculate it. \par
\begin{figure}[h]
\begin{center}
\begin{tikzpicture}

\draw[thick] (-3, -1) -- (-1, 0) -- (1, -1) -- (1, -2) -- (1, -3);
\draw[thick] (-1, 0) -- (-1, -1);
\draw[thick] (-3, -1) -- (-3, -2) -- (-3, -3);
\filldraw[gray] (-1,0) circle (2pt);
\node[] at (-1,.3) {$\emptyset$};
\filldraw[gray] (-3, -1) circle (2pt);
\node[] at (-3.4, -1) {$\langle \lambda \rangle$};
\filldraw[gray] (-1, -1) circle (2pt);
\node[] at (-.6, -1) {$\langle F \rangle $};
\filldraw[gray] (1, -1) circle (2pt);
\node[] at (1.4, -1) {$\langle T \rangle $};
\filldraw[gray] (-3, -2) circle (2pt);
\node[] at (-2.4, -2) {$\langle \lambda, T \rangle$};
\filldraw[gray] (1, -2) circle (2pt);
\node[] at (1.6, -2) {$\langle T, T \rangle$};
\filldraw[gray] (-3, -3) circle (2pt);
\filldraw[gray] (1, -3) circle (2pt);
\node[] at (-2.15, -3) {$\langle \lambda, T, F\rangle $};
\node[] at (1.85, -3) {$\langle T, T, F\rangle$};

\draw[] (-5.5, .5) -- (-6, .5) -- (-6, -2.3) -- (-5.5, -2.3);
\draw[] (-5.5, -2.7) -- (-6, -2.7) -- (-6, -3.0);
\draw[] (-6, -3.0) -- (-6, -4.5);
\draw[] (-6, -4.5) -- (-5.5, -4.5);

\node[rotate=90] at (-6.3, -1) {Cluster 0};
\node[rotate=90] at (-6.3, -3.6) {Cluster 1};
\node[] at (-5.3, 0) {Layer 0};
\node[] at (-5.3, -1) {Layer 1};
\node[] at (-5.3, -2) {Layer 2};
\node[] at (-5.3, -3) {Layer 3};
\node[] at (-5.3, -4) {Layer 4};

\draw[blue] (-3.2, -3.2) -- (-3.2, -2.8) -- (-2.8, -2.8) -- (-2.8, -3.2) -- cycle;
\draw[blue] (-1.2, -1.2) -- (-1.2, -.8) -- (-.8, -.8) -- (-.8, -1.2) -- cycle;
\draw[blue] (.8, -3.2) -- (.8, -2.8) -- (1.2, -2.8) -- (1.2, -3.2) -- cycle;

\draw[ultra thick, green] (-3.2, -2.8) -- (-3, -3) -- (-2.7, -2.5);
\draw[ultra thick, green] (-1.2, -.8) -- (-1, -1) -- (-.7, -.5);

\end{tikzpicture}

\caption{The \SIMD\ operation, shown as the associated clusters, with each box marked by a blue box. Each layer corresponds to a variable (with Layer 4 being empty because all boxes have only three variables), and checkmarks mark the boxes that are found to be containing boxes. The central branch is created from the box $\langle F, \lambda, \lambda\rangle$. In the left branch, created by $\langle \lambda, T, F\rangle$, a child is created corresponding to the sub-box $\langle \lambda, T\rangle$, and then the box is inserted in the child cluster at $\langle \lambda, T, F\rangle$. In the right branch, created by $\langle T, T, F\rangle$, we similarly create a child corresponding to the sub-box $\langle T, T\rangle$, and then inserting $\langle T, T, F\rangle$ in the child cluster. We have drawn the database here as a 3-trie; to see it as a single, flattened cluster, see Figure \ref{cluster_flat}.}
\label{simd_picture}
\end{center}
\end{figure}
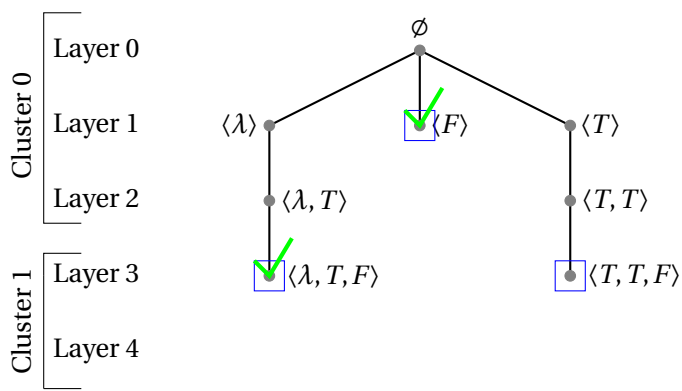

\begin{figure}[h]
\begin{center}
\begin{tikzpicture}[scale=1]

\node[] at (-2.90, .05) {Database's \boxes : \texttt{\ 0 0 1 0 0 0 0 0 0 0 0 0 0}};
\draw[] (-3.85, .25) -- (1.35, .25) -- (1.35, -.15) -- (-3.85, -.15) -- cycle;
\foreach \x in {1,...,12}
{
   \draw[] (-3.85 + \x * .4, .25) -- (-3.85 + \x * .4, -.15);
}

\node[] at (-1.3, -.9) {\huge \textbf{\&}};

\node[] at (-2.41, -1.6) {Input: $\langle F, T \rangle$: \texttt{\ 1 1 1 0 1 1 0 1 0 1 0 0 0}};
\draw[] (-3.85, .-1.35) -- (1.35, -1.35) -- (1.35, -1.75) -- (-3.85, -1.75) -- cycle;
\foreach \x in {1,...,12}
{
   \draw[] (-3.85 + \x * .4, -1.35) -- (-3.85 + \x * .4, -1.75);
}

\node[] at (-1.3, -2.7) {\huge \textbf{=}};

\node[] at (-2.8, -3.2) {Containing Boxes: \texttt{\ 0 0 1 0 0 0 0 0 0 0 0 0 0}};
\draw[] (-3.85, .-2.95) -- (1.35, -2.95) -- (1.35, -3.35) -- (-3.85, -3.35) -- cycle;
\foreach \x in {1,...,12}
{
   \draw[] (-3.85 + \x * .4, -2.95) -- (-3.85 + \x * .4, -3.35);
}
\foreach \x in {0, ..., 12}
{
\node[] at (-3.65 + \x * .4, -.4) {\x};

\node[] at (-3.65 + \x * .4, -2.0) {\x};

\node[] at (-3.65 + \x * .4, -3.6) {\x};
}

\end{tikzpicture}
\caption{The algorithm takes the logical AND of the stored bitstrings, top, with the bitstring corresponding to the input, row 2, in order to produce the list of containing boxes and children, bottom. The \boxes\ bitstring contains a single 1 to mark the box $\langle F, \lambda \rangle$. The second line contains all the potential boxes that would contain $\langle F, T\rangle$; in addition to the boxes that do exist, these include $\langle \lambda \rangle$, $\langle F, T\rangle$, and many more. The output of the operation has bits set corresponding to the box $\langle F, \lambda \rangle$.}
\label{simd_bitstring}
\end{center}
\end{figure}
     Let us inspect our output. In this particular example, we find that we have matched the box $\langle F, \lambda \rangle$ and the child with prefix $\langle \lambda, T \rangle$. \par
    At this point, we are faced with an interesting choice. Suppose for the sake of argument that the child eventually leads to some containing box for the input. Now, if this is a check for containment, the algorithm can only return the one box that it considers the ``best" containing box. Which one, then, should the algorithm choose? In general, it will choose the shortest-length box; in this example, it will choose the box $\langle F, \lambda \rangle$. The reason for this is simple: the more $\lambda$s in a box, the more space it covers; and the more space it covers, the more output space it can cover. Additionally, when those $\lambda$s are at the tail end  of a box, we know that we can immediately advance the probe point considerably. On the other hand, when the $\lambda$s are in the middle, the algorithm must re-find these boxes every time it scans a point that is contained within this hypothetical box. This repetition is costly; we would rather avoid it.\par
    Now, let us consider how our algorithm would handle this input if we were simply using traditional tries; in other words, consider what would happen if our cluster size were 1, as in a standard 3-trie. This algorithm would query the input box for its first value, find that it is $F$, and know that it could check both the $\lambda$ branch and the $F$ branch. Since $\lambda$s are generally to be preferred, it would take the $\lambda$ branch. Then, it would match the $T$ value to the $T$ branch, and proceed to the third layer, where it would match the $F$ value to the $F$ branch and find a containing box, which it will return. In this way, it has taken us three comparisons to find a containing box. Furthermore, the box found in the 3-trie version is of lower quality: we would have preferred to find the $\langle F, \lambda, \lambda \rangle$ box instead. \par
    On the other hand, let us consider what occurs when using full, four-variable clusters (Figure \ref{cluster_flat}). In this case, a single operation immediately finds all three containing boxes. Therefore, we have outperformed the traditional variation both in terms of number of comparisons and in terms of the quality of the output. \par
    In summary, each box is stored as a single bit in a 128-bit vector (with the last 7 bits unused), as is a record of whether or not a given child exists. A lookup table is used to find the 128-bit vectors corresponding to the possible outputs. Then, the former two are concatenated, as are the latter two, and all are compared using a single 256-bit \sfAND\ operation. It can be seen that the output of this operation must be the intersection of the potential containing boxes and children, found from the lookup table, and the ones that actually exist. Hence,  by calculating $\phi (b)$ for some box $b$, we can quickly find a box containing $b$, if it exists; and if it does not, we can quickly generate the exact list of children to examine.\par
    Additionally, in practice, it turns out that there is a certain sub-box value that shows up far more often than any other: the sequence $\langle\lambda ,\lambda ,\lambda ,\lambda \rangle$. Notably, this is the very sequence that accepts every single possible input string. Therefore, in the event that a layer contains only this child and no words, it is in fact possible to skip the entire layer. In practice, this produces significant savings on both computational costs and memory usage, which contributes towards Theoretical Implication \ref{implications:time-space}.\par

Let us now formally define the data structures and its algorithms used in our implementation:

\begin{defn}[Data Structure]
The data structure is a 121-ary trie, where each node of the trie is called a {\em cluster}. The top level consists of a pointer to the root cluster, which is the cluster covering the first four variables in the ordering, and can perform the following operations: \Insert\ (Algorithm \ref{algo_insert}), \Contains\ (Algorithm \ref{algo_contains}), and \GetAllContainingBoxes\ (Algorithm \ref{algo_gacb}).
\end{defn}

\begin{defn}[Cluster]
Each cluster in the database contains two bitstrings, $\boxes$  and $\children$, bitstrings identifying the sets of boxes and children, respectively, along with an integer $\depth$ that informs the cluster of its depth. The operation $.at(i)$ on a bitstring calculates the intersection of that bitstring with a bitstring retrieved from a lookup table that lists the set of boxes or children that contain or potentially contain, respectively, the box with value $i$; setting $.at(i)$ to 1 sets the specific bit referring exactly to that box. Each cluster corresponds to four layers of a standard 3-ary trie.
\end{defn}

\begin{figure}[h]
\begin{center}
\begin{tikzpicture}

\draw[thick] (-3, 0) -- (-6, -1);
\draw[thick] (-3, 0) -- (-4, -1);
\draw[thick] (-3, 0) -- (-1, -1);

\filldraw[gray] (-3,0) circle (2pt);
\node[] at (-3,.3) {$\emptyset$};
\filldraw[gray] (-6, -1) circle (2pt);
\node[] at (-6.5, -1) {$\langle F \rangle$};
\filldraw[gray] (-4, -1) circle (2pt);
\node[] at (-3, -1) {$\langle \lambda, T, F\rangle$};
\filldraw[gray] (-1, -1) circle (2pt);
\node[] at (0, -1) {$\langle  T, T, F\rangle$};

\draw[] (-6.5, .5) -- (-7, .5) -- (-7, -1.5) -- (-6.5, -1.5);

\node[rotate=90] at (-7.3, -.6) {Cluster 0};

\end{tikzpicture}
\caption{The clusters in Figure \ref{simd_picture}, flattened out into a single node with four variables to a cluster. This is how the node would be stored in the actual database. Note that it is far more compressed than when drawn out as a trie as in Figure~\ref{simd_picture}. This cluster would have a \depth\ of 0, have three bits set in \boxes, and no bits set in \children.}
\label{cluster_flat}
\end{center}
\end{figure}
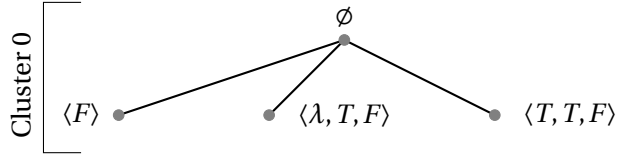

\begin{defn}[Index of a Box]
The index of a box $b$, $index(b)$, is the location of the last non-$\lambda$ variable in that box.
\end{defn}
\begin{example}
The index of the box $b = \langle F, T, \lambda, F, \lambda \rangle$, $index(b)$, is 4.
\end{example}

\begin{algorithm}[H]
\caption{$\Insert$}
\label{algo_insert}
\begin{algorithmic}[1]
\STATE{Given box $b$ to insert,  $b = \langle c_1, c_2, ..., c_n \rangle,$ where each $c_i$ here is a cluster consisting of $b_{i_1}...b_{i_4}$}
\STATE{$k \gets index(b)$}
\STATE{Call \InsertCluster($b, k$) (Algorithm \ref{algo_insert_cluster}) called on the root cluster}
\end{algorithmic}
\end{algorithm}

\begin{algorithm}[H]
\caption{\InsertCluster}
\label{algo_insert_cluster}
\begin{algorithmic}[1]
\STATE{Input: A cluster $T$ with depth \depth to insert on, box $b$ to insert,  $b = \langle c_1, c_2, ..., c_n \rangle,$ where each $c_i$ here is a cluster consisting of $b_{i_1}, ..., b_{i_4}$, and the location of the final non-$\lambda$ cluster k}
\IF{$\depth= k$}
\IF{$\boxes.at(\phi(c_{\depth}))$ is nonempty}
\STATE{Return} \COMMENT{\texttt{If a containing box of $b$ is already in the data structure, stop.}}
\ELSE
\STATE{Set \boxes.at$(\phi(c_{\depth}))$ to 1}
\ENDIF
\ELSE
\IF{$\children.at(\phi(c_{\depth})) \ne 1 $}
\STATE{Create a child cluster at location $\phi(c_{\depth})$ with depth $(\depth + 1)$.}
\STATE{Set $\children.at(\phi(c_{\depth}))$ to 1}
\ENDIF
\STATE{Call \InsertCluster\ on the cluster indexed as $\phi(c_{\depth})$.}\
\ENDIF
\STATE{Return}
\end{algorithmic}
\end{algorithm}

In the \Insert\ algorithms, Algorithms \ref{algo_insert} and \ref{algo_insert_cluster}, we recursively traverse clusters until we find the appropriate location in the data structure, and then set the appropriate bit to 1. Along the way, we will check for containing boxes and immediately cease operation if one is found. Furthermore, if a child cluster that contains the box we are inserting, or that contains a cluster along the path to that cluster, does not exist, we create it. \par

\begin{algorithm}[H]
\caption{$\Contains$}
\label{algo_contains}
\begin{algorithmic}[1]
\STATE{Given box $b$ to find a containing box of,  $b = \langle c_1, c_2, ..., c_n \rangle,$ where each $c_i$ here is a cluster consisting of $b_{i_1}...b_{i_4}$}
\STATE{$output \gets \ContainsCluster(b)$ (Algorithm \ref{algo_contains_cluster}) called on the root cluster}
\STATE{Return $output$}
\end{algorithmic}
\end{algorithm}

\begin{algorithm}[H]
\caption{$\ContainsCluster$}
\label{algo_contains_cluster}
\begin{algorithmic}[1]
\STATE{Input: Cluster $T$ with depth \depth that is being checked for a containing box, box $b$ to check containment of,  $b = \langle c_1, c_2, ..., c_n \rangle,$ where each $c_i$ here is a cluster consisting of $b_{i_1}...b_{i_4}$.}
\IF{($f = \boxes.at(\phi(c_{\depth})))$ is nonempty (\texttt{i.e., there is at least one box in the intersection})} 
\STATE{$o = min_{x \in f} Index(x)$} 
\STATE{Return $o$}
\ELSIF {($f = \children.at(\phi(c_{\depth})))$ is nonempty (\texttt{i.e., there is at least one child in the intersection})}
\FORALL{Children $k \in f$}
\STATE {$a = k.\ContainsCluster(b)$} \COMMENT{\texttt{Scan these in order of increasing index}}
\IF{$a$ is nonempty}
\STATE {Return $a$}
\ENDIF
\ENDFOR
\ELSE
\STATE{Return $\emptyset$}
\ENDIF
\end{algorithmic}
\end{algorithm}

In the \Contains\ algorithms, Algorithms \ref{algo_contains} and \ref{algo_contains_cluster}, we check the database to see if it contains any box that contains some box $b$. Therefore, we traverse along clusters in our path, first checking to see if any containing boxes exist; if we find one, we return that box immediately and cease checking further. If none exists, we will perform a depth-first search of the children that could potentially contain a containing box. This process continues until either a containing box is found, or else the search space is exhausted and it is determined that no containing box exists. \par

\begin{algorithm}[H]
\caption{$\GetAllContainingBoxes$}
\label{algo_gacb}
\begin{algorithmic}[1]
\STATE{Given box $b$ to find all containing boxes of,  $b = \langle c_1, c_2, ..., c_n \rangle,$ where each $c_i$ here is a cluster consisting of $b_{i_1}...b_{i_4}$}
\STATE{$output \gets \GetAllContainingBoxesCluster(b)$ (Algorithm \ref{algo_gacb_cluster}) called on the root cluster}
\STATE{Return $output$}
\end{algorithmic}
\end{algorithm}

\begin{algorithm}[H]
\caption{$\GetAllContainingBoxesCluster$}
\label{algo_gacb_cluster}
\begin{algorithmic}[1]
\STATE{Input: Cluster $T$ with depth \depth\ that is being checked, box $b$ to find all containing boxes of,  $b = \langle c_1, c_2, ..., c_n \rangle,$ where each $c_i$ here is a cluster consisting of $b_{i_1}...b_{i_4}$.}
\STATE{$O \gets \emptyset$}
\IF{($F = \boxes.at(\phi(c_{\depth})))$ is nonempty  (\texttt{i.e., there is at least one box in the intersection})} 
\STATE{$O = F$.} 
\ELSIF {($f = \children.at(\phi(c_{\depth})))$ is nonempty (\texttt{i.e., there is at least one child in the intersection})}
\FORALL{Children $k \in f$}
\STATE {$A = k.\GetAllContainingBoxesCluster(b)$}
\STATE {$O = O \cup A$}
\ENDFOR
\ELSE
\STATE{Return $O$}
\ENDIF
\end{algorithmic}
\end{algorithm}

The \GetAllContainingBoxes\ algorithms, Algorithms \ref{algo_gacb} and \ref{algo_gacb_cluster}, are similar to the \Contains\ algorithms. There are, however, two key differences. First, while \Contains\ terminates as soon as it found a single containing box, \GetAllContainingBoxes\ will continue. Secondly, it returns the set of all containing boxes, rather than just one; hence, the name. In all other regards, it behaves exactly as \Contains\ does. \par

\subsection{Global Variable Ordering}
\label{sec:gvo}
\label{SEC:GVO}
Up until this point, we have simply been assuming that all boxes must order their variables in exactly the same order as they appear in the original \SAT\ formulas. In other words, in each box $b$, $b_1$ must correspond to $x_1$, $b_2$ must correspond to $x_2$, and so on. However, this does not need to be the case. We can reorder the variables, and by doing so can greatly improve the runtime of our system.\par
The original \Tetris\ paper cites the importance of the variable ordering. However, it assumes that there exists an exponential in $n$ time algorithm to compute the optimal variable ordering. While this is justifiable in the context of join problems, where $n$ is small compared to the size of the database, in \SAT\ problems it is unacceptable. Furthermore, as computing the optimal ordering is NP-hard, we cannot hope to improve upon this result. Indeed, even approximating the ordering is intractable~\cite{marx}.\par
    Nevertheless, initial experimentation with \Tetris\ made clear how impactful this choice can be. A slight variation in ordering can result in a large difference in runtime. Thus, we turned to various heuristics and intuitions in order to find a quick and effective means to generate an ordering that works well in practice and thereby contribute to Theoretical Implication \ref{implications:ordering}.\par

    First, let us define a few terms that we will use in our discussion of various ordering strategies:
\begin{defn}[Degree]
The degree of a variable is the number of clauses of which the variable is a part.
\end{defn}
\begin{example}
In Example \ref{exprob1}, $(x_1 \vee x_2) \wedge (x_1 \vee \bar{x_2}) \wedge (x_2 \vee x_3)$, $x_1$ has degree 2, $x_2$ has degree 3, and $x_3$ has degree 1.
\end{example}

\begin{defn}[Closeness]
Variables $x_i$ and $x_j$ are said to be {\em close} if there exists a clause that includes both $x_i$ and $x_j$. The fewer terms in the clause, the closer the two variables are said to be. Specifically, the closeness of two variables, $\theta (x_1, x_2)$, is equal to 1 divided by the size of the smallest clause containing both variables minus 1.
\end{defn}
\begin{example}
In the clause $(x_1 \vee x_2)$, $x_1$ and $x_2$ would have a closeness of 1; and in the clause $(x_1 \vee x_2 \vee x_3)$, $x_1$ and $x_3$ would have a closeness of $\frac{1}{2}$. If both clauses were part of the same \SAT\ problem, $x_1$ and $x_2$ would still have a closeness of 1 because the first clause has a smaller size than the second.
\label{closeness_example}
\end{example}

\begin{defn}[Interconnectedness]
The interconnectedness of a cluster $C$, $IC(C)$, is the sum of the closeness values for each $\binom{4}{2}$ pairs of variables in the cluster.
\end{defn}
\begin{example}
Using Example \ref{closeness_example}, if $x_1, x_2,$ and $x_3$ compose a cluster, its interconnectedness would be 2, since $\theta (x_1, x_2) = 1$, $\theta(x_1, x_3) = \frac{1}{2},$ and $\theta(x_2, x_3) = \frac{1}{2}$.
\end{example}

    In general, we note two high-level strategies that we have found improve the performance of a given global variable ordering:
    \paragraph{a) High-degree First} In \Tetris, handling high-degree variables early tends to improve performance. To see the reason for this, we must consider the nature of the algorithm. Scattering high-degree variables  throughout the ordering forces the algorithm to branch frequently, which means that, when testing for inclusion or for containing boxes, the algorithm must scan all possible branches. This is highly inefficient. Instead, we focus the branches as much as possible to the beginning, with the hope being that, as the algorithm progresses down from there, most layers will have few if any divergent choices. If this is true, then the inclusion check can be handled quickly. Let us proceed to see why this is so.
    First, we will introduce an example of why placing high-degree variables early in the ordering proves effective. Let us consider the following sample problem:
\begin{example} Consider the \SAT\ formula
$(x_1 \vee x_3) \wedge (x_2 \vee x_3)$. The equivalent box problem, using the ordering $(x_1, x_2, x_3)$, would begin with database $D$ containing the boxes $\{\langle F,\lambda , F\rangle, \langle\lambda, F, F\rangle\}$. $x_1$ has degree 1, $x_2$ has degree 1, and $x_3$ has degree 2.
\end{example}
     Now, let us consider how the algorithm would attack this problem if we used this naive, low-degree first ordering, which would be $(x_1, x_2, x_3)$. We immediately note two things. First, there are no boxes that have $\lambda$ for the final variable. This means that the algorithm will never find a box that allows it to skip multiple probe points unless it can use resolution to create a new box that happens to have that property. In fact, that will not occur. Additionally, we can consider all 8 possible probe points and track how many comparisons the algorithm would need to make on each point, assuming that each cluster only covers a single variable instead of four. We can see that the algorithm will always have to calculate the set intersection for the cluster of depth 1, will have to perform the set intersection for the cluster with depth 2 on 50\% of probe points, and will have to perform the set intersection for the cluster with depth 3 on 75\% of probe points. \par
    Let us contrast this with the high-degree ordering, which moves $x_3$ to the front and $x_1$ to the back to give us the ordering $(x_3, x_2, x_1)$. Now, our database contains the boxes $\langle F, \lambda, F \rangle$ and $\langle F, F, \lambda \rangle$.  This time, we do have a box with a $\lambda$ at the back; specifically, $\langle F, F, \lambda \rangle$. Therefore, after the probe point $\langle F, F, F\rangle$ finds this box, the algorithm will advance past $\langle F, F, T\rangle$ entirely. Additionally, while we still have to perform a set intersection for the cluster with depth 1 100\% of the time, and 50\% of the time for the cluster with depth 2, we only have to perform the set intersection at depth 3 25\% of the time --- and all of these numbers ignore how we skipped one of the probe points entirely. \par
    While this is of course a very simple example, this illustrates the principles that cause the strategy to be effective in larger datasets.

   \paragraph{ b) Local Interconnectedness} As a direct result of the 121-ary trie-based system described in Section \ref{Data_Structure}, if a box has multiple non-$\lambda$ variables within the same 4-variable cluster, they can all be recovered with a single operation. Therefore, maximizing the interconnectedness of these 4-variable blocks provides an advantage.\par
     To illustrate, let us consider the following example:
\begin{example}
$(x_1 \vee x_3) \wedge (x_2 \vee x_4)$, with each cluster containing 2 variables rather than 4. 
\label{prob3}
\end{example}
    If we use the naive strategy of keeping the variables ordered as-is, the first cluster contains $x_1$ and $x_2$ while the second cluster contains $x_3$ and $x_4$. Therefore, both clusters have interconnectedness 0, and we find that all boxes transcend a cluster boundary.  In other words, it will always take at least two comparisons to find either of these boxes. However, if we had gone with the ordering $(x_1, x_3, x_2, x_4)$ instead, the box corresponding to $(x_1 \vee x_3)$ would be entirely contained within the first cluster, and the box corresponding to $(x_2 \vee x_4)$ would be entirely contained within the second cluster. Therefore, each box would be entirely contained within a cluster, and each cluster would have had an interconnectedness of 1, and each box could have been recovered with only a single comparison. This saves a large number of comparisons over the long term.

\subsubsection{Ordering Algorithms}
\label{OrderingAlgos}
    Two major methods were employed in order to achieve these aims. The first was a descending degree sort. While this only directly achieved goal (a), in practice it did an acceptable job with goal (b). Additionally, we constructed three variations on this method. The first, naive degree descent, is where we simply order the variables according to their degree, using Algorithm \ref{ndd}. \par 

\begin{algorithm}[H]
\caption{Naive Degree Descent Ordering}
\begin{algorithmic}[1]
\STATE{Given a set of variables $V$ and, for each variable $v$, its degree $v_d$:}
\STATE{$O \gets $ SORT($V$ on $v_d$, descending)}
\STATE{Return $O$}
\end{algorithmic}
\label{ndd}

\end{algorithm}

   The second, optimally grouped degree descent, forms all possible groups of four variables, finds the greatest possible interconnectivity among these groups, and then selects from all groups with the greatest interconnectivity on the basis of the combined degree of the group of four, using Algorithm \ref{ogdd}. While this proved effective, it is a slow algorithm with a runtime of $\Theta(n^8)$.\par

\begin{algorithm}[H]
\caption{Optimally Grouped Degree Descent Ordering}
\begin{algorithmic}[1]
\STATE{Given a set of variables $V$ and, for each variable $v$, its degree $v_d$:}
\STATE{$O \gets \emptyset$}
\STATE{$G \gets $ all possible sets of four variables $\{v_i, v_j, v_k, v_l\}$ from $V$.} 
\WHILE{$G$ is nonempty}
\STATE{$maxIC \gets max_{g \in G}IC(g)$} \COMMENT{\texttt{Determine the maximum possible interconnectedness of all remaining groups.}}
\STATE{$X \gets max_{g \in G~s.t.~IC(g) = maxIC}(\sum_{v \in g}v_d)$} \COMMENT{\texttt{Of the groups with max interconnectedness, select the group for which the sum of the degrees of all variables is the greatest.}}
\STATE{$O \gets (O, X)$} \COMMENT{\texttt{Append this group to the ordering.}}
\FOR {$v \in X$}
\FOR {$Y \in G$}
\IF {$v \in Y$}
\STATE {$G \gets G \setminus Y$} \COMMENT{\texttt{Remove each grouping that contains one of the variables in the group that was selected.}}
\ENDIF
\ENDFOR
\ENDFOR
\ENDWHILE
\STATE{Return $O$}
\end{algorithmic}
\label{ogdd}
\end{algorithm}

    This necessitated the creation of the third subtype, the heuristically grouped degree descent ordering. This ordering works in groups of four. When creating a group, the first node chosen is the highest-degree remaining variable. Then, for each of the remaining three variables, the algorithm picks the variable with the highest interconnectedness to the nodes already chosen for this group of four, breaking inevitable ties based on degree. The result, Algorithm \ref{hgg}, is an algorithm that can compute its ordering significantly faster than the optimal ordering, while \Tetris\ run on this ordering runs is competitive with the optimal ordering.\par

\begin{algorithm}[H]
\caption{Heuristically Grouped Degree Descent Ordering}
\begin{algorithmic}[1]
\STATE{Given a set of variables $V$ and, for each variable $v$, its degree $d_v$:}
\STATE{$i \gets 0$}
\STATE{$X \gets \emptyset $}
\STATE{$O \gets \emptyset $}
\WHILE{$V$ is nonempty}
\IF {$i = 0$} 
\STATE {$x \gets y $ where $ Degree(y) = max_{v \in V}(d_v)$}
\ELSE
\STATE{$maxIC \gets max_{v \in V}(\sum_{x \in X} \theta(v, x))$} \COMMENT{\texttt{Calculate which variable has the best interconnectedness with the already chosen variables}}
\STATE{$x \gets y $ where $d_y = max_{v \in V s.t. V_{IC} = maxIC}(d_v)$} \COMMENT{\texttt{Break ties based on degree}}
\ENDIF

\STATE{$O \gets (O, x)$}
\STATE{$X \gets X \cup x$}
\STATE{$i \gets i + 1$}
\IF{$i = 3$}
\STATE{$i = 0$}
\STATE{$x \gets \emptyset$}
\ENDIF
\ENDWHILE
\STATE{Return $O$}
\end{algorithmic}
\label{hgg}
\end{algorithm}

Additionally, we employ the Treewidth tree decomposition, which was introduced in \cite{treewidth}. In essence, the idea here is to minimize the width of the search tree; in our domain, this corresponds to increasing the locality and local interconnectedness of variables. This naturally did a very good job with interconnectivity, while a decent job with placing high-degree variables early.\par
    We also experimented with the \textsf{minfill} ordering, described in \cite{minfill}. This ordering sets the elimination order such that the node to be eliminated is the node whose removal makes the smallest impact on the overall graph. While this ordering has proved effective in similar applications, we found it to perform poorly with \Tetris.\par
    In Table \ref{orderingTable}, we can see how these various orderings performed in practice on representative graph-based and non-graph-based benchmarks.  For instance, while the Treewidth sort outperformed all others on the AIS8 dataset \cite{AIS}, on the WikiVotes dataset (created using a \SNAP\ \cite{SNAP} dataset; see Section \ref{dataset_generation}) the ordering caused \Tetris\ to timeout. Notably, we see that the Heuristically Grouped Degree Descent takes only slightly longer to process the input compared to the Naive Degree Descent, but significantly less time than the Optimally Grouped Degree Descent; however, the runtime does not suffer significantly when going from the optimal ordering to the heuristic one. \par

\begin{table}[h]
\begin{center}

\begin{tabular}{|c|c|c|c|}

\hline
\multicolumn{4}{|c|}{\textsc{Runtime with Different Orderings}}\\
\hline 
Dataset & Ordering & Load Time (Seconds) & Runtime (Seconds)\\
\hline
\multirow{5}{*}{AIS8} & Naive Degree Descent & .001 & 1.494\\
                                     & Heuristically Grouped Degree Descent & .009 & 2.607\\
                                     & \textbf{ Treewidth} & \textbf{.002} & \textbf{.45}\\
                                     & Minfill & .008 & 7.005\\
                                     & Optimally Grouped Degree Descent & 1.629 & 0.573\\
\hline
\multirow{5}{*}{WikiVotes} & Naive Degree Descent & .927 & 34.124\\
                                     & \textbf{ Heuristically Grouped Degree Descent }&  \textbf{1.024 }& \textbf{ 21.509}\\
                                     & Treewidth & 2.349 & Timeout\\
                                     & Minfill & 2.032 & 4059.426\\
                                     & Optimally Grouped Degree Descent & 2.06 & 23.142\\

\hline

\end{tabular}

\caption{The performance of various ordering schemes on two datasets. As one can see, no ordering does best on both datasets; indeed, the best ordering on one is the worst on the other. Insertion Ratio (see Section \ref{Selective_Insertion}) was set to .5 for these tests.}
\label{orderingTable}
\end{center}
\end{table}

\subsection{Selective Insertion}
\label{Selective_Insertion}
\label{SEC:INSERT}
While the original \Tetris\ paper calls for the insertion of every box created through the resolution process to be inserted into the database, this proved to be inefficient in practice. Very frequently, this will result in a huge increase in the number of branches that the algorithm must scan while trying to find the output point without notably improving the quality of the containing boxes found. Therefore, we only insert those boxes that contain a suitably high percentage of $\lambda$s. The best results generally come from requiring slightly less than 50 percent of the layers to be composed of entirely $\lambda$s.\par
    In Table \ref{InsertTable}, we have posted the runtime for the AIS10 dataset \cite{AIS} showing the relationship between the number of $\lambda$s we require in storing a box and the runtime of \Tetris. Performance suffers at extreme settings, with optimal performance resulting from an insertion ratio of close to $\frac{1}{2}$. Hence, with regard to Theoretical Implication \ref{implications:time-space}, we find that by decreasing space complexity, we furthermore improve runtime. \par

\begin{table}[h]

\begin{center}
\begin{tabular}{|c|c|}

\hline
\multicolumn{2}{|c|}{\textsc{Insertion Ratio vs. Runtime}}\\
\hline 
Ratio & Time (Seconds)\\
\hline
.00 & 140.999\\
.05 & 123.49\\
.10 & 79.003\\
.15 & 71.229\\
.20 & 41.956\\
.25 & 36.685\\
.30 & 32.064\\
.35 & 24.405\\
.40 & 22.749\\
\textbf{.45} & \textbf{21.784}\\
.50 & 24.171\\
.55 & 25.205\\
.60 & 28.157\\
.65 & 35.738\\
.70 & 44.657\\
.75 & 51.809\\
.80 & 54.622\\
.85 & 66.529\\
.90 & 64.16\\
.95 & 75.829\\
1.00 & 88.001\\
\hline
\end{tabular}
\caption{A comparison of insertion ratios with the time to solve the AIS10 dataset \cite{AIS}. For these tests, we used the Treewidth ordering; similar behavior was observed from all ordering strategies.}
\label{InsertTable}
\end{center}

\end{table}

\section{Our Experimental Results}
\label{sec:res}
\label{SEC:RESULTS}
Here, we compare \CNFTetris\ --- that is, \Tetris\ designed to solve \CNF\ problems --- with other model counters in order to compare and contrast their ability to tackle model counting problems. A model counting problem, simply put, is, given a \CNF\ formula, output the number of satisfying solutions to that formula. Since \Tetris\ was originally designed to handle database joins, these are more natural problems for the algorithm to solve than the corresponding \SAT\ problems, which are to simply determine whether or not any solution exists. \par
While the model counting problem allows for a solver to simply find the number of solutions without finding the solutions themselves, \CNFTetris\ does in fact output all of the solutions. While this admittedly poses a disadvantage compared to the solvers we are comparing against, for some datasets, \CNFTetris\ runs faster in spite of this. \par
We compare our results with those of the \SharpSATs\ \cite{sharpsat} , \dSharp \cite{muise2012dsharp}, and \Cachet \cite{sang2004combining} model counters due to their recognition as state-of-the-art model counters. All tests were performed using a single thread on an 8-core E5 v3 2.6GHz processor with 64GB of RAM. \par
   Additionally, we include two types of datasets. The first is derived from join problems on graphs. These are the sort of problems that \Tetris\ was originally designed to solve; as such, \CNFTetris\ does a very good job on them. The second set is a selection of standard model counting benchmarks from various competitions held over the past several years. Most model counters have been trained to solve such problems, so they serve as an apt second set of benchmarks for \CNFTetris\ to compete against.

\subsection{Graph Results}
Here, we will compare and contrast how various solvers performed on model counting problems created from graphs.
\begin{table*}[ht!]
\begin{center}
\begin{adjustwidth}{-1.25cm}{}

\begin{tabular}{|c|c|c|c|c|c|c|c|c|c|}
\hline
\multicolumn{10}{|c|}{\textsc{Runtime on Various Datasets}}\\
\hline
Query & Base Graph & \multicolumn{2}{|c|}{\CNFTetris} & \multicolumn{2}{|c|}{\SharpSATs} & \multicolumn{2}{|c|}{\dSharp} & \multicolumn{2}{|c|}{\Cachet}\\
\hline
& & Loadtime & Runtime & Runtime & Speedup & Runtime & Speedup & Runtime & Speedup \\
\hline
\multirow{3}{*}{3-clique} & Wikivotes & 1.080 & \textbf{21.55} & Timeout & n/a & Timeout & n/a & Timeout & n/a \\
& Facebook & .593 & \textbf{11.22} & Timeout & n/a & Timeout & n/a & Timeout & n/a \\
& Soc-Epinions & 77.420 & \textbf{309.418} & Timeout & n/a & Timeout & n/a & Timeout & n/a \\
\hline
\multirow{3}{*}{2-path} & Wikivotes & 2.32 & \textbf{35.171} & 31801.8 & 904.2 & 36947.5 & 1050.51 & 28838.1 & 819.94 \\
& Facebook & .752 & \textbf{10.236} & 4637.25 & 453.03 & 3871.86 & 378.26 & 2210.42 & 215.95 \\
& Soc-Epinions & 6.22 & \textbf{236.439} & Timeout & n/a & Timeout & n/a & Timeout & n/a \\
\hline
\hline
\multicolumn{2}{|c|}{AIS6} & 0.0021 & .013 & \textbf{.004} & .308 & .0074 & .569 & .0135 & 1.034 \\
\hline
\multicolumn{2}{|c|}{AIS8} & 0.0027 & .45 & \textbf{.0466} & 0.104 & .3386 & 0.752 & .300 & .667 \\
\hline
\multicolumn{2}{|c|}{AIS10} & .00798 & 21.55 & \textbf{2.629} & .122 & 14.757 & .685 & 16.4724 & .764 \\
\hline
\multicolumn{2}{|c|}{AIS12} & .0198 & 1732.03 & \textbf{124.937} & .0721 & 1002.73 & .579 & 1020.64 & .589 \\
\hline
\multicolumn{2}{|c|}{ls8-simplified4} & .00063 & .123 & \textbf{.004933} & .0401 & .021 & .171 & .01 & .0813 \\
\hline
\multicolumn{2}{|c|}{LS5 firstr} & .000656 & .409 & \textbf{.0156} & .0381 & .095 & .232 & .025 & .0611 \\
\hline
\end{tabular}
\end{adjustwidth}

\caption{This table shows the comparative results of various solvers on \CNF\ datasets created using various \SNAP\ graphical datasets and \SAT\ datasets. All runtimes are in seconds; timeout was set at 40,000 seconds. For these tests, we used a insertion ratio of .45 and the Heuristic Degree Descent ordering for \CNFTetris. Wikivote \cite{SNAP} contains 39 variables and 745485 clauses; Facebook \cite{SNAP} contains 36 variables and 464234 clauses; and Soc-Epinions \cite{SNAP} contains 51 variables and 4578589 clauses (all clause data is for 3-clique; 2-path has approximately 2/3rd that number of clauses). For the SAT datasets, AIS6 \cite{AIS} has 61 variables and 581 clauses; AIS8 \cite{AIS} has 113 variables and 1520 clauses; AIS10 \cite{AIS} has 181 variables and 3151 clauses; AIS12 \cite{AIS} has 265 variables and 5666 clauses; ls8-simplified4 \cite{MC} has 119 variables and 410 clauses; and LS5-firstr \cite{MC} has 125 variables and 529 clauses. In \CNFTetris, loadtime refers to the time to determine the variable ordering and insert the boxes into the database, while runtime is the time to find all satisfying solutions.}
\label{table:main}

\end{center}
\end{table*}

\subsubsection{Dataset Generation}
\label{dataset_generation}
The \CNF\ graph datasets were created using the publicly available \SNAP\ datasets \cite{SNAP}. These are graph datasets; that is, each consists of a set of vertices and a set of edges connecting those vertices. Each of these datasets is a natural problem; some arose from social networks, while others are anonymized data from other corners on the Internet. We can then use this data to run various queries; for instance, we can determine how many triangles exist in the graph. Our goal, then, is to convert these problems into an equivalent \CNF\ problem so that we can use \CNFTetris\ and the other model counters to solve them.\par
To do this, each vertex is first assigned a unique binary encoding using $\log(n)$ bits. We furthermore increase the number of bits such that there are $\log(n)$ bits times the size of the data structure being looked for in the graph. For instance, if we are performing the triangle query on the dataset, there will be $3 \cdot \log(n)$ bits used in the encoding. Henceforth, let $k$ represent the size of this query. Each of these bits will correspond to a variable in the \CNF\ encoding of the problem. In essence, each of these repetitions represents a vertex in e.g. the triangle.\par
    Next, we will encode each absent edge (i.e., a pair of vertices $v_1$ and $v_2$ such that the edge $(v_1, v_2) \notin E$, where $E$ is the edge set of the graph) as $\binom{k}{2}$ total Boolean formulas for the $k$-clique query, and $ k$ formulas for the $k$-path query. Each of these formulas corresponds to e.g. one of the three edges of a triangle. Observe that any possible satisfying solution to the \SAT\ problem cannot select an edge that does not exist; therefore, any assignment that matches one of these formulas on all variables must be rejected. Equivalently, any accepting assignment must match at least one variable in the inverse.  This naturally leads to a \CNF\ definition, which is what we create.  We will repeat this encoding over each possible set of vertex pairings such that the lower-indexed vertex is always written before the higher-indexed vertex, while adding additional clauses to reject all edges that would be from the higher-indexed vertex to the lower-indexed vertex. Simplifying resolutions are also performed where possible. Therefore, we have created a \CNF\ problem where, for each output to the query in the original problem, there exists a solution. We can and do use this problem instance as input for both \Tetris\ and other model counters. 
   Let us examine an example instance of this. Consider the very simple example graph depicted in Figure \ref{dataset_gen_example} below, using the triangle query. In order to encode the non-edge $(v_2, v_4)$, we must first calculate the binary encodings of each of these vertices. These are (01) and (11), respectively. We then flip all of the bits, giving us (10) and (00). Then we construct three \CNF\ clauses, each of which corresponds to being the first, the second, or the third edge of the triangle. The first will be $(x_1 \vee \bar{x_2} \vee \bar{x_3} \vee \bar{x_4})$, with $(x_1 \vee \bar{x_2})$ corresponding to (10) and $(\bar{x_3} \vee \bar{x_4})$ corresponding to 00. Similarly, the second will be $(x_1 \vee \bar{x_2} \vee \bar{x_5} \vee \bar{x_6})$; and the third will be $(x_3 \vee \bar{x_4} \vee \bar{x_5} \vee \bar{x_6})$. Additionally, we insert clauses forbidding ``bad" orderings of the points; in other words, we are making sure we do not count $(v_1, v_2, v_3)$ and $(v_2, v_3, v_1)$ as separate triangles. \par
    Then, when \Tetris\ run on this \CNF\ input attempts to recover the number of triangles, when it uses the probe point (assuming naive ordering) $\langle T, T, T, F, F, T \rangle$ --- that is, the probe point that corresponds to the inverted binary representations of $v_1, v_2,$ and $v_3$, the three vertices in the top-right triangle --- let us consider what happens. Since each of the selected edges does not correspond to the missing edge, we know that all three of those clauses must be satisfied; and since each edge is in the index order, we know that the additional clauses that we added will also accept our input. Therefore, \Tetris\ will add this probe point to the output list. This continues for all other probe points until \Tetris\ has found all triangles.

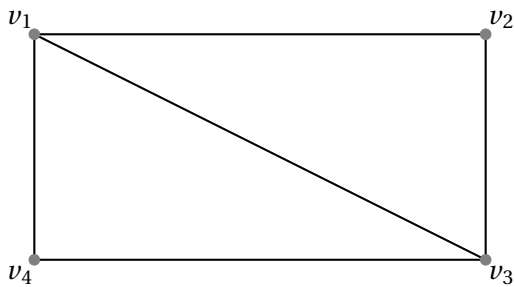
\begin{figure}[H]
\begin{center}
\begin{tikzpicture}

\draw[black, thick] (-3,0) -- (3, 0) -- (3, 3) -- (-3, 3) -- cycle;
\draw[black, thick] (-3,3) -- (3, 0);
\filldraw[gray] (3,0) circle (2pt);
\filldraw[gray] (-3,0) circle (2pt);
\filldraw[gray] (3,3) circle (2pt);
\filldraw[gray] (-3,3) circle (2pt);

\node[] at (-3.2, 3.2) {$v_1$};
\node[] at (3.2, 3.2) {$v_2$};
\node[] at (3.2, -.2) {$v_3$};
\node[] at (-3.2, -.2) {$v_4$};

\end{tikzpicture}
\caption{A sample graph on four vertices, used in the above example. We will encode the non-edge $(v_2, v_4)$ as our \SAT\ formula, along with additional clauses that ensure we only count triangles once, which will allow us to run the query as a \SharpSAT\ problem.}
\label{dataset_gen_example}
\end{center}
\end{figure}

\subsubsection{Results Analysis}
As can be seen in Table~\ref{table:main}, while all of the other solvers find these problems to be difficult, \CNFTetris\ solves them quickly. Queries that take seconds on \CNFTetris\ wind up taking hours on the competition, with \CNFTetris\ running nearly a thousand times faster on some problems. This is largely due to the extremely high number of clauses relative to the number of variables, along with the fact that these clauses contain a large number of variables; these factors are not present in many of the standard \SAT\ benchmarks. For instance, while the average clause in many \SAT\ benchmarks contains two or three variables, here the average clause has thirty or more. And while \SAT\ benchmarks rarely have over ten times as many clauses as variables, here the system is forced to tackle an environment where the number of clauses is exponentially larger than the number of variables. Note that all of the solvers we are comparing against use unit propagation techniques in order to count models \cite{SATsurvey}; see Section \ref{sec:rw} for details. Because of this, the increased number of clauses directly corresponds to increased work for these solvers. 

\subsection{Nongraph Results}
In this section, we will discuss how \CNFTetris\ performed as compared to other solvers on standard model counting benchmarks.
\subsubsection{About the Datasets}
These datasets are a combination of datasets from the SATLIB datasets \cite{AIS} and the SampleCount benchmarks for model counting taken from International Joint Conference on Artificial Intelligence '07 \cite{MC}. We chose to use the AIS datasets for several reasons. First, each of the datasets terminates in a reasonable amount of time on all solvers, allowing us to find interesting comparisons. Secondly, due to the existence of increasingly-sized versions of this dataset, we can use this as insight into whether or not \Tetris\ is scaling efficiently with the size of the dataset. Additionally, we featured the ls8-simplified4 and LS5firstr datasets, which gave us insights into our implementation's strengths and weaknesses.

\subsubsection{Results Analysis}
As Table \ref{table:main} shows, \Tetris\ is competitive with \dSharp\ and \Cachet\ on many of the datasets. Indeed, a factor of 2 separates us from either solver on all of the AIS datasets, a difference that engineering work alone can easily overcome. While there is significantly more space between it and \SharpSATs, a factor of 10 on average, we believe the distance is not insurmountable. \par
    The largest gaps exist on the ls8-simplified and LS5-firstr datasets. On these, \CNFTetris\ is roughly a factor of 5 off of the worst of the competition, and a factor of over 20 as compared to \SharpSATs. The reason for this is simple: both of these datasets contain pure variables. In other words, there exist variables $x$ such that, in all clauses, $\bar{x}$ never appears, or vice-versa. This is an important piece of information, one that can and must be utilized, but \CNFTetris\ in its current state does not know how to do so. However, since we know what the problem is, we expect to be able to quickly and efficiently attack this issue.

\section{Related Work}
\label{sec:rw}
 Our work builds on \Tetris\ as developed by Abo Khamis et al. in \cite{AboKhamis:2015}. In that work, the authors introduced \Tetris\ as a beyond-worst-case algorithm for geometrically solving the database join problem.  This in turn built on work on the Minesweeper \cite{Minesweeper}, NPRR \cite{NPRR} and Leapfrog \cite{Leapfrog} algorithms, of which \Tetris\ is a generalization. Furthermore, \Tetris\ itself is can be considered a version of the \DPLL\ algorithm \cite{DPLL} with clause learning. In \DPLL, which is itself an evolution of the earlier DP\ \cite{DP} algorithm, a variable is chosen at every stage and assigned to be either true or false. The algorithm then uses unit propagation in order to simplify clauses under these assumptions. In this techniques, after the solver assigns a value to a variable, every other clause is inspected to see if this assignment creates a unit clause (i.e. a clause with only one variable in it), and to see if resolutions can be performed. This process continues until a conflicting clause (that is, a clause that is violated by the assignments) is found, at which point the algorithm is forced to backtrack. In the clause learning versions, introduced in \cite{silva1997grasp}, the solver takes this as an opportunity. It determines where it went astray, adds a new clause to its cache that is the negation of this errant assignment, non-chronologically backtracks to where this decision-making took place, and then proceeds in the opposite direction. \par
   The reasoning why \CNFTetris\ is a form of this algorithm follows from the aforementioned method of converting from \SAT\ clauses to boxes, and vice-versa (see Algorithm \ref{algo:convert}). Since these two representations are exactly equivalent, any operation performed on one representation can be translated into an operation on the other. Hence, every single operation \Tetris\ performs on the boxes over its execution must correspond exactly to a set of operations on the original clauses. \par
    For instance, the \Contains\ operation matches up with the idea of a conflicting clause. When a containing box is found, we can consider this as finding a box that rejects the current probe point as a potential output point. Meanwhile, a conflicting clause rejects a potential satisfying assignment in much the same way. Furthermore, over the course of \Tetris, the algorithm tentatively assigns a variable to either true or false, and then proceeds along with this assumption until a contradiction is found, all while learning additional clauses where possible through the resolution process. When a containing box is found or synthesized through resolution, and we advance the probe point accordingly, we are in essence backtracking to the earliest decision point and choosing to go in the opposite direction, just as \DPLL\ with clause learning does. Therefore, this is exactly the \DPLL\  algorithm with clause learning, with the added restriction of a fixed global variable ordering \cite{AboKhamis:2015}. \par
     \Tetris\ additionally utilizes a three-value logic system. While similar systems have been utilized in database schemes, such as by Zaniolo in \cite{Zaniolo:1982}, in these systems the three values are \true, \false, and \unknown. Here, however, the three values we are considering can be summarized as \true, \false, and \sfboth. This causes a number of key differences. For instance, $\true \wedge \unknown$ is equivalent to \unknown, while $\true \wedge \sfboth$ is equivalent to \true. Similarly, $\true \vee \unknown$ is equivalent to \true, while $\true \vee \sfboth$ is equivalent to \sfboth. \par
     Much work has been done in creating \SAT\ solvers. Let us briefly discuss those state-of-the-art solvers we are comparing our work against. First, let us consider \Cachet \cite{sang2004combining}. This solver was originally released in 2005, with minor compatibility updates continuing through the most recent version, which came out in 2015 \cite{cachet_page}. \par
Next, there is \SharpSATs. First released in 2006, \SharpSATs\ significantly eclipsed contemporary solvers \cite{sharpsat}. \SharpSATs\ has been maintained over time, with the most recent release in 2013 \cite{sharpsat_page}. \par
    Finally, we come to \dSharp. The most recently released of our three competitors, \dSharp was introduced in 2012 in order to efficiently compile CNF problems into the Decomposable Negation Normal Form language \cite{muise2012dsharp}. Further work allowed it to function as a model counter, which is how we utilize it. The version we use was released in 2016. \par
    What all of these solvers have in common, including \CNFTetris, is that they have at their core a form of the \DPLL\ algorithm with clause learning; indeed, almost all modern \SAT\ and \SharpSAT\ solvers do so \cite{SATsurvey}. The differences, then, come in terms of efficiency. Each solver uses a different array of techniques in order to effectively cache and recover learned clauses, to determine the variable ordering, and to identify clause conflicts. \par
    With \Cachet, the authors focused on adding component caching capabilities on top of an existing \SAT\ solver, ZChaff \cite{sang2004combining}, the theoretical grounds for which were themselves introduced in \cite{zchaff}. This caching involved the storing of subproblems in a local cache, so that these clauses would not have to be re-derived by \Cachet\ at a later juncture, thereby reducing redundant calculations over the course of the algorithm. This can be viewed as analogous to how \CNFTetris\ stores learned boxes in a local cache, which it checks for containing boxes before examining the original database. A subproblem, meanwhile, could be thought of as a box with a high percentage of variables set to $\lambda$. \par
   However, one key difference here is the nature of the cached components. In \Cachet, due to how the algorithm functions, it must regularly prune the cache of siblings that would otherwise cause it to undercount the number of models \cite{sang2004combining}. \CNFTetris, in contrast, needs to perform no such pruning; it will naturally determine the exact number of models without any additional work. \par
   \SharpSATs\ built on the work in \Cachet\ while adding new ideas of its own \cite{sharpsat}. Boolean constraint propagation (also known as the failed literal rule \cite{GSS09}) and unit propagation heuristics are used by \SharpSATs\ to identify failed literals with greater efficiency than was done in \Cachet\ \cite{sharpsat}. However, by fixing the variable order, \CNFTetris\  simplifies this process. Ultimately, this means that it finds its conflicting boxes in a fundamentally different manner than \SharpSATs\ does, which provides room for \CNFTetris\ to outperform \SharpSATs. \par
   \dSharp, much like how \SharpSATs\ built on \Cachet\, uses \SharpSATs\ as a core component \cite{muise2012dsharp}. The authors perform a DNNF translation, and then use properties of decomposability and determinism to perform model counting \cite{GSS09}. Though these differences do allow it to outperform more pure DPLL-based solvers on some benchmarks \cite{GSS09}, since this system still uses \SharpSATs\ as a core component, it still shares many of the same advantages and disadvantages in comparison to \CNFTetris. \par
   As we have seen, all of the competing solvers can be viewed as evolutions along a single line. While \CNFTetris\ does not throw the baby out with the bathwater --- that is, while \CNFTetris\ still continues to implement the classic \DPLL\ algorithm --- it does represent a distinct deviation from that line, challenging assumptions such as the necessity of allowing a non-fixed global variable ordering and the much more complex data storage scheme necessary in order to accommodate this. While this has necessitated much work in order to implement, it has also shown vast promise.

\section{Acknowledgments}
We would like to thank Mahmoud Abo Khamis, Hung Q. Ngo, Christopher R\'{e}, and Ce Zhang for very helpful discussions. 

\bibliographystyle{plain}
\bibliography{ref}

\end{document}